\documentclass[reqno, 12pt]{article}         

 \usepackage[T1]{fontenc}              
 \usepackage[latin1]{inputenc}         
      \input{dehyphtex}                
 \usepackage[width=16cm, height=22cm]{geometry}   
 \usepackage{amsmath,amssymb}          
\usepackage{hyperref} 

\usepackage[amsmath,thmmarks,hyperref]{ntheorem} 
 \usepackage{icomma}                   
 \usepackage{enumerate}                
 \usepackage{color}                    
 \usepackage{setspace}                 
 \usepackage{needspace}                
 \usepackage{url}                      
 \usepackage{mathrsfs}                 
 \usepackage{stmaryrd}                 
 \usepackage{graphicx}                 
 \usepackage{tikz}                     
\usepackage[arrow, matrix, curve]{xy}

\usepackage{slashed}

 \theoremstyle{nonumberplain}  
 \theoremheaderfont{\itshape}  
 \theorembodyfont{\normalfont}  
 \theoremseparator{.}  
 \theoremsymbol{$\Box$}
 \newtheorem{proof}{Proof} 
 
 \theoremstyle{plain}  
 \theoremheaderfont{\normalfont\bfseries}  
 \theorembodyfont{\itshape}  
 \theoremsymbol{~}
 \theoremseparator{.}  
 \newtheorem{proposition}{Proposition}[section]  
\newtheorem{corollary}[proposition]{Corollary}  
\newtheorem{lemma}[proposition]{Lemma}  
\newtheorem{theorem}[proposition]{Theorem}   

 \theorembodyfont{\normalfont}

\newtheorem{remark}[proposition]{Remark}

\newtheorem{definition}[proposition]{Definition} 
\newtheorem{setup}[proposition]{Setup} 

 \theoremstyle{nonumberplain}
 \theorembodyfont{\normalfont}  

\DeclareMathOperator{\diag}{diag}

\newcommand*{\grad}{\operatorname{grad}}

\newcommand{\R}{\mathbb{R}}
\newcommand{\Eh}{\mathcal{E}}

\newcommand{\N}{\mathbb{N}}
\newcommand{\Z}{\mathbb{Z}}
\newcommand{\C}{\mathbb{C}}
\newcommand{\D}{\mathscr{D}}

\newcommand{\dd}{\mathrm{d}}

\newcommand{\End}{\mathrm{End}}

\newcommand{\lau}{(\!(\hbar^{1/2})\!)}
\newcommand{\Ka}{\mathfrak{S}}
\newcommand{\<}{\left\langle}
\renewcommand{\>}{\right\rangle}

\newcommand{\spec}{\mathrm{spec}}

\newcommand{\natop}[2]{\genfrac{}{}{0pt}{}{#1}{#2}}

\numberwithin{equation}{section}

\title{Asymptotic eigenfunctions for Schr\"odinger operators on a vector bundle}
\author{Matthias Ludewig \and Elke Rosenberger}

 \date{\today}

\begin{document}

 \maketitle

\begin{center}
Universit\"at Potsdam / Institut f\"ur Mathematik \\ 
  Am Neuen Palais 10 / 14469 Potsdam, Germany \\ \medskip
elke.rosenberger@uni-potsdam.de \\
\vspace{0.5cm}
  Max-Planck-Institut f\"ur Mathematik \\
  Vivatgasse 7 / 53111 Bonn \\ \medskip
 maludewi@mpim-bonn.mpg.de
\end{center}

\vspace{0.5cm}

  \begin{abstract}
  In the limit $\hbar\to 0$, we analyze a class of Schr\"odinger operators $H_\hbar = \hbar^2 L + \hbar W + V\cdot \mathrm{id}_\Eh$ acting on sections 
of a vector 
bundle $\Eh$ over a Riemannian manifold $M$ where $L$ is a Laplace type operator, $W$ is an endomorphism field
  and the potential energy $V$ has a non-degenerate minimum at some point $p\in M$.
  We construct quasimodes of WKB-type near $p$ 
  for eigenfunctions associated with the low lying eigenvalues of $H_\hbar$. These are obtained
  from eigenfunctions of the associated harmonic oscillator $H_{p, \hbar}$ at $p$, acting on smooth functions on the tangent space.
  \end{abstract}

\section{Introduction}

In this paper, we study semi-classical quasimodes of WKB-type, formally associated with the low lying spectrum of a
Schr\"odinger operator $H_\hbar$ on a vector bundle $\Eh$ over a smooth Riemannian manifold $M$. More precisely, given an operator of the form
\begin{equation*}
   H_\hbar = \hbar^2 L + \hbar W + V \cdot \mathrm{id}_\Eh
\end{equation*}
acting on the space $\Gamma^\infty(M, \Eh)$ of smooth sections of $\Eh$, we construct formal asymptotic eigenfunctions near non-degenerate minima of the potential 
$V$ in the limit $\hbar \to 0$. 

Operators of this type arise e.g.\ in Witten's perturbation of the de Rham complex where $H_\hbar$ is the square of 
the Dirac type operator
\begin{equation*}
  \hbar \, e^{\phi/\hbar} \bigl( \dd + \dd^* \bigr) e^{-\phi/\hbar} = \hbar\bigl( \dd + \dd^* \bigr) + \dd \phi \wedge + \dd \phi 
\lrcorner\, .
\end{equation*}
In this particular case, the endomorphism $W$ is non-vanishing, which is the reason for us to include this (somewhat unusual) term in our considerations.

We recall that the construction of semi-classical quasimodes of WKB-type is an important step in discussing tunneling problems, i.e.\
exponentially small splitting of eigenvalues for a self-adjoint realization of $H_\hbar$. In the scalar case, for $\dim M>1$, 
rigorous results in this field start with the seminal paper \cite{helffer-sjostrand-1} (for $M=\R^n$ or
$M$ compact). The associated asymptotic expansion of eigenvalues was also considered in \cite{simon-1}, and somewhat
weaker results on tunneling were obtained in \cite{simon-2} avoiding the use of quasimodes of WKB-type.

For non-scalar operators, the bundle of exterior differential forms (for $M=\R^n$ or $M$ compact) has been considered in  \cite{helffer-sjostrand-4} in the
context of the Witten complex. All WKB-constructions in \cite{helffer-sjostrand-1}
as well as in \cite{helffer-sjostrand-4} ultimately rely on the asymptotic constructions done in \cite{helffer-sjostrand-1} which proceed via 
a special FBI-transform. There exist several introductory texts to this field (e.g.\ \cite{dima},  \cite{helffer},
\cite{helffer2}) but none of these treats the case of eigenvalues degenerate in the harmonic approximation (in this
case, the naive
WKB-constructions which work for the non-degenerate eigenvalues, break down, see Remark 2.3.5 in \cite{helffer}).
However, for the scalar case and $M=\R^n$, there exists a more elementary approach to the asymptotic WKB-constructions, avoiding the 
use of FBI-transform (see \cite{klein-schwarz}). This approach was also used in \cite{klein-rosen1}
for the case of semi-classical difference operators on the scaled lattice $\hbar \Z^n$.

The central point of this paper (see Thm.~\ref{Theorem1} and Corollary~\ref{CorollaryTheorem1} below) is to show that this method gives complete asymptotic solutions of WKB-type for a class of
Schr\"odinger operators on general bundles. 

Since our results are local, we do not need further restrictions on $M$ (as e.g.\ compactness, completeness, bounded geometry)
and not even a self-adjoint realization of $H_\hbar$. In particular, we do not discuss the tunneling problem for $H_\hbar$ with a 
multiwell potential $V$ thus avoiding the use of Agmon-type estimates for the true eigenfunctions of certain Dirichlet operators
and estimates on the difference between WKB-type quasimodes and these eigenfunctions far from the well and with exponential precision.

\medskip

\textbf{Acknowledgements}. The first author is indebted to SFB 647 and the Max Planck Institute for Mathematics in Bonn for financial support.

\section{Outline of the Results}

In everything what follows, let $(M, g)$ be a (smooth) Riemannian manifold and let $\Eh$ be a complex vector bundle over $M$ equipped with an inner 
product $\gamma$ (i.e.\ a positive definite Hermitian form). This gives a unique volume density inducing an integral $\int_M$ for compactly supported 
continuous functions. 
The standard inner product on 
$\Gamma^\infty_c(M, \Eh)$, the space of compactly supported smooth sections of $\Eh$, is then defined by
\begin{equation}\label{standard_skalar}
  (u, v)_{\gamma} = \int_M \gamma[ u, v ]~~~~~~ u, v \in \Gamma_c^\infty(M, \Eh)\, .
\end{equation}

Recall that a differential operator $L$ acting on sections of $\Eh$ is said to be of {\em Laplace type} if, in local coordinates $x$, it has the form
\begin{equation}\label{Laplace-type}
  L = - \mathrm{id}_{\Eh} \sum_{ij} g^{ij}(x)\frac{\partial^2}{\partial x^i\partial x^j} + \sum_{j} 
b_j \frac{\partial}{\partial x^j} + c
\end{equation}
where $\bigl(g^{ij}(x)\bigr)$ is the inverse matrix of the metric $\bigl(g_{ij}(x)\bigr)$ and $b_j, c\in \Gamma^\infty (M, \End (\Eh))$ are
endomorphism fields. We always assume that $L$ is {\em symmetric}, also called {\em formally self-adjoint}, which means that
\begin{equation}\label{Lsymm}
  \int_M \gamma[ L u, v ] = \int_M \gamma[ u, Lv ] ~~~~~~~~ \text{for all} ~~~ u, v \in \Gamma_c^\infty(M, \Eh).
\end{equation}
Examples of symmetric Laplace-type operators are the Hodge-Laplacian on 
$p$-forms (in particular, for $p=0$, this is the Laplace-Beltrami operator) and the square of a generalized Dirac operator acting on spinors.

\begin{remark} \label{RemarkLaplaceType}
Given any symmetric Laplace type operator $L$ on a vector bundle $\Eh$, there exists a metric connection 
$\nabla^\Eh$ on $\Eh$ and a symmetric endomorphism field $K \in \Gamma^\infty(M, \End(\Eh))$ such that
\begin{equation}\label{L-Darstellung}
  L = (\nabla^\Eh)^* \nabla^\Eh + K
\end{equation}
(see e.g. in \cite[Prop.\ 2.5]{bgv}). In fact, $\nabla^\Eh$ and $K$ are uniquely determined. Namely, if there is another metric connection $\widetilde{\nabla}^\Eh$ and an endomorphism field $\widetilde{K}$ such that
$L=(\widetilde{\nabla }^\Eh)^* \widetilde{\nabla}^\Eh + \widetilde{K}$, then we have $\widetilde{\nabla}^\Eh = \nabla^\Eh + A$ for some 
$A\in C^\infty(M, T^* M\otimes \text{asymEnd}(\Eh))$. Inserting this into the equation for $L$ and expanding gives $K = A^*\nabla^\Eh + (\nabla^\Eh)^*A + \widetilde{K}$, thus the right
hand side has to be of order zero. Since $D := A^*\nabla^\Eh + (\nabla^\Eh)^*A + \widetilde{K}$ has the principal symbol $\sigma_1(D,\xi)=-2 A_{\xi^\#}$, it follows that 
$A=0$ and thus $\widetilde{\nabla}^\Eh = \nabla^\Eh$. 
\end{remark}

Our setup is the following.

\begin{setup}\label{setup1}
For $\hbar>0$, we consider Schr\"odinger operators $H_\hbar$ acting on $\Gamma^\infty(M, \Eh)$
of the form
\begin{equation}\label{schrodinger}
  H_\hbar = \hbar^2 L + \hbar W +V \cdot \mathrm{id}_\Eh
\end{equation}
where $L$ is a symmetric Laplace type operator as above, 
$W \in \Gamma^\infty(M, \End(\Eh))$ is a symmetric endomorphism field and $V \in C^\infty(M, \R)$. Furthermore, we assume that the potential $V$ has a 
non-degenerate minimum at some fixed point $p\in M$ with $V(p) = 0$.
\end{setup}

We remark that the operator $H_\hbar$ given in \eqref{schrodinger} is not necessarily real, i.e.\ it does not commute with complex conjugation. 
However, under
semi-classical quantization ($\xi \mapsto -i\hbar \dd$ in some reasonable sense) its principal $\hbar$-symbol
\begin{equation}\label{symbol}
 \sigma_H(q, \xi) = \bigl( |\xi|^2 + V(q)\bigr) \mathrm{id}_\Eh
\end{equation}
is both real and scalar ($|\cdot|$ denotes the norm on $T_q^*M$ induced by $g$). This is crucial for our construction.
Thus our assumptions exclude Schr\"odinger operators with magnetic field (the operator $(i\hbar d + \alpha)^*(i\hbar d + \alpha)$, with a $1$-form 
$\alpha$ describing the magnetic potential, has non real
principal $\hbar$-symbol, see e.g.\ \cite{helffer-kondyukov-1}) or with endomorphism valued potential $V$ as needed e.g.\ for molecular Hamiltonians in the
Born-Oppenheimer approximation (see e.g.\ \cite{klein-seiler}).

\begin{definition}[Local Harmonic Oscillator] \label{DefLocalHarmonicOscillator}
In Setup \ref{setup1}, we associate to $H_\hbar$ the {\em local harmonic oscillator} $H_{p,\hbar}$ at the critical point $p$ of $V$. This is the differential 
operator acting on the space $C^\infty(T_p M, \Eh_p)$ of smooth $\mathcal{E}_p$-valued functions on $T_p M$ by
\begin{equation}
  H_{p,\hbar} f(X) := \Bigl( \hbar^2 \Delta_{T_p M} + \hbar \, W(p) + \frac{1}{2} \nabla^2V|_p (X, X)\Bigr) f(X),
\end{equation}
for $X \in T_pM$, where $\Delta_{T_p M}$ denotes the (flat) Laplacian on $T_p M$ induced by the scalar product $g_p$ on $T_pM$ and $\nabla^2 V|_p$ denotes the 
Hessian of $V$ at $p$. The latter is a bilinear form, so that the operator $\nabla^2V|_p (X, X)$ acts by multiplication with a quadratic function.
\end{definition}

The connection between $H_\hbar$ and the associated harmonic oscillator $H_{p,\hbar}$ will be discussed in detail in Section \ref{Kapitel3}, in particular in 
Remark \ref{RemarkOnH0AndQ0}.

The local harmonic operator $H_{p,\hbar}$ can be considered as an essentially self-adjoint operator on 
$C^\infty_c(T_p M, \Eh_p)\cap L^2(T_pM, \Eh_p)$ (where the $L^2$ space is taken with respect to the Lebesgue measure induced by the scalar product on $T_pM$). It is well known that its spectrum scales with $\hbar$
and consists of the numbers
\begin{equation} \label{LocalEigenvalues}
  \hbar E_{\alpha,\ell} = \hbar\bigl((2\alpha_1 +1) \lambda_1 + \dots + (2\alpha_n + 1) \lambda_n + \mu_\ell\bigr),  ~~~~~ \alpha \in \N_0^n, ~~ \ell=1, \dots, 
\mathrm{rk}\,\Eh
\end{equation}
where $\lambda_1, \dots, \lambda_n$ are the eigenvalues of $\frac{1}{2}\nabla^2 V|_p$ and $\mu_1, \dots, \mu_{\mathrm{rk} \Eh}$ are the eigenvalues 
of $W(p)$ (see e.g. \cite{cycon}, \cite[Section 8.10]{reed-simon-1}). 

\begin{remark}\label{polynom}
It is clear that $H_{p,\hbar}$ maps polynomials to polynomials, i.e., it preserves $\Eh_p[T_p M] := \Eh_p \otimes \C[T_p M] \subset 
C^\infty(T_p M, \Eh_p)$, the space of polynomials on $T_p M$ with values in $\Eh_p$.
\end{remark}

To formulate our results, we need the following theorem (a proof can be found for example in \cite{helffer-sjostrand-1}).

\begin{theorem}[The Eikonal Equation]\label{ThmEikonal} 
In Setup \ref{setup1}, for each sufficiently small neighborhood $U$ of $p$, there exists a unique function 
$\phi\in C^\infty(U, \R)$ with $\phi(p) = 0$ that has a non-degenerate minimum at $p$ and satisfies the eikonal equation\footnote{We remark that the eikonal 
equation can be written as $\sigma_H(p, \xi) = 0$, where $\sigma$ denotes the semi-classical principal symbol, 
compare \eqref{symbol} above. This is important in the case of more general operators.}
\begin{equation} \label{EikonalEquation}
 \bigl|\dd \phi(q)\bigr|^2 = V(q), ~~~~~ q \in U,
\end{equation}
where $|\cdot|$ denotes the norm on $T_q^*M$ induced by $g_q$.
\end{theorem}

\begin{definition} \label{DefAdmissible}
Let $U \subseteq M$ be an open neighborhood of $p$ and $\phi \in C^\infty(U, \R)$. We call $(U, \phi)$ an {\em admissible pair} (with respect to
$H_\hbar$), if
\begin{enumerate}
\item[1.)] $\phi$ is the unique positive solution of the eikonal equation \eqref{EikonalEquation} on $U$ as in Theorem \ref{ThmEikonal}.
\item[2.)] $U$ is star-shaped around $p$ with respect to the vector field $\grad \phi$ in the following sense: If $\Phi_t$ is the flow of $\grad \phi$, then we have $\Phi_t(U) \subseteq U$ for all $t \leq 0$.
\end{enumerate}
\end{definition}

Our main result is the following theorem, stating that for each eigenvalue of the local harmonic oscillator at $p$, we can obtain asymptotic eigenvectors up to any order in $\hbar$.

\begin{theorem} \label{Theorem1}
In Setup \ref{setup1}, let $(U, \phi)$ be an admissible pair and let $\hbar E_0$ be an eigenvalue of multiplicity $m_0$ of the local harmonic oscillator 
$H_{p, \hbar}$ at $p$ as given in Def.\ \ref{DefLocalHarmonicOscillator}. Define the operator  $H_{\phi, \hbar}$ over $U$ by
\begin{equation*}
H_{\phi, \hbar} u = e^{\phi/\hbar} \circ H_\hbar [ e^{-\phi/\hbar} u].
\end{equation*}
Then there exist a number $K \in \N_0/2$ and formal power series
\begin{equation*}
  \boldsymbol{a}_j = \hbar^{- K}\sum\nolimits_{k\in \N_0/2} \hbar^k a_{j,k}, ~~~~~\text{and}~~~~~ 
\boldsymbol{E}_j = \hbar\bigl(E_0 + \sum\nolimits_{k\in \N/2} \hbar^k E_{j,k}\bigr), 
\end{equation*}
for $j=1, \dots, m_0$, where $a_{j,k} \in \Gamma^\infty(U, \Eh)$ and $E_{j,k} \in \R$, such that in the sense of formal power series 
\begin{enumerate}
\item[1.] we have
\begin{equation*}
  H_{\phi, \hbar} \boldsymbol{a}_j = \boldsymbol{E}_j \boldsymbol{a}_j\, , \qquad j=1, \dots, m_0\; .
\end{equation*}
 \item[2.] the series $\boldsymbol{a}_1, \dots, \boldsymbol{a}_{m_0}$ are asymptotically orthonormal meaning that
\begin{equation*}
 \mathcal{I}\Bigl( \gamma \bigl[ \chi \boldsymbol{a}_j, \chi \boldsymbol{a}_i\bigr]\Bigr) = \delta_{ji} \, , \qquad i, j=1, \dots, m_0,
 \end{equation*}
 as asymptotic series in $\hbar^{1/2}$. Here, $\chi \in \Gamma_c^\infty (U, [0,1])$ is any cutoff function with $\chi \equiv 1$ in a 
 neigborhood of $p$ and $\mathcal{I}(\gamma[\cdot, \cdot])$ is the weighted integral defined in \eqref{extendI}.
\end{enumerate}
Furthermore, if $|\alpha|$ is even  (or odd respectively) in all pairs $(\alpha, \ell)$ such that $E_0 = E_{\alpha, \ell}$ defined in \eqref{LocalEigenvalues} then no 
half-integer terms (or integer terms respectively) occur in the series $ \boldsymbol{a}_j $. In both cases, 
no half integer terms occur in the series $\boldsymbol{E}_j$\footnote{Results concerning the parity were first proven in 
\cite{helffer-sjostrand-1}, correcting a mistake in \cite{simon-1}, see \cite{simon-1-e}.}.
\end{theorem}

\begin{remark}
The lowest order in $\hbar$ in the expansion of $\boldsymbol{a}_j$ is given by $K = \max_\alpha |\alpha|/2$ where $\alpha$ runs over all 
multi-indices such that $E_{\alpha,\ell} = E_0$ for some $\ell=1, \dots \mathrm{rk}\,\Eh$.
\end{remark}

\begin{remark} \label{ExplainingThm1}
Property 1 and 2 in Thm. \ref{Theorem1} are equivalent to the following statements: 
For each $U^\prime \subset \subset U$ open and compactly contained in $U$, each $N \in \N/2$ and each $\hbar_0>0$, there exists constants 
$C_1, C_2>0$ such that for each $\hbar < \hbar_0$ and all $1 \leq i, j \leq m_0$ we have
\begin{equation} \label{Asymptotics}
  \Bigl| \Bigl( H_\hbar - \hbar\bigl(E_0 - \sum\nolimits_{k=1/2}^N \hbar^k E_{j,k}\bigr) \Bigr) e^{-\phi/\hbar} \sum\nolimits_{k=0}^N \hbar^k a_{j,k} \Bigr| 
\leq C_1 \,e^{-\phi/\hbar} \hbar^{K+N+ 3/2}
\end{equation}
uniformly on $U^\prime$ for each $j = 1, \dots, m_0$ and
 \begin{equation} \label{AsymptoticOrthonormal}
 \int_{U^\prime} \gamma \Bigl[e^{-\phi/\hbar}\sum\nolimits_{k=0}^N \hbar^k a_{i,k}~, ~e^{-\phi/\hbar}\sum\nolimits_{k=0}^N \hbar^k a_{j,k} \Bigr]
\leq \hbar^{2K+n/2}(\delta_{ij} + C_2 \hbar^{N+1/2}),
\end{equation}
where in each case, the sums are meant to run in half integer steps. Heuristically, this means that the series $e^{-\phi/\hbar} \boldsymbol{a}_j$, $j=1, \dots, m_0$ form an asymptotically orthonormal basis of the eigenspace to the asymptotic eigenvalue $\mathbf{E}_j$.
\end{remark}

Theorem \ref{Theorem1} allows us to construct good quasimodes for $H_\hbar$, which are essential in the discussion of tunneling problems.
These are derived by a Borel procedure with respect to $\hbar^{1/2}$.

\begin{corollary} \label{CorollaryTheorem1}
Under the assumptions of Thm.\ \ref{Theorem1},
for any open neighborhood $U^\prime \subset \subset U$ of $p$, there are functions\footnote{for asymptotic series with coefficients in topological vector 
spaces, compare \cite[Thm.\ 1.2.6]{hormander}.} $a_j \in C^\infty((0, \hbar_0), \Gamma^\infty(M, \Eh))$ and $E_j \in C^\infty((0, \hbar_0), \R)$, 
$j = 1, \dots, m_0$, 
such that $a_j(\hbar)$ is compactly supported in $U$ for each $\hbar< \hbar_0$ and 
\begin{equation*}
 a_j (\hbar) \sim \hbar^{-n/4} e^{-\phi/\hbar}\boldsymbol{a}_j\quad\text{and}\quad E_j(\hbar) \sim \boldsymbol{E}_j 
\end{equation*}
on $U^\prime$ as $\hbar \searrow 0$ for
$\boldsymbol{a}_j$ and $\boldsymbol{E}_j$ given in Thm. \ref{Theorem1}. 
Moreover, we have
\begin{equation}
H_\hbar a_j(\hbar) = E_j(\hbar) a_j(\hbar) + o(\hbar^\infty)
\end{equation}
uniformly on $M$ and
\begin{equation}
H_\hbar a_j(\hbar) = E_j(\hbar) a_j(\hbar) + o(\hbar^\infty e^{-\phi/\hbar})
\end{equation}
uniformly on $U^\prime$, as well as the asymptotic orthonormality relation
\begin{equation}
  \bigl(a_i, a_j\bigr)_\gamma = \delta_{ij} + o(\hbar^\infty)
\end{equation}
for the inner product \eqref{standard_skalar}.
\end{corollary}

\begin{remark}
We do not make any claim that the quasimodes constructed above are in any sense asymptotic to actual eigenfunctions of the Schr\"odinger operator 
\eqref{schrodinger}. In fact, this statement does not make sense as it stands, as one would need to specify a self-adjoint realization of the operator on a 
suitable Hilbert space, which need not even be unique without further assumptions on the manifold and the operator. Even though statements that 
formal asymptotics actually belong to eigenfunctions hold in great generality, this is not in the scope of our paper. We refer to \cite{helffer-sjostrand-1} 
for a discussion of $-\Delta$ on functions in $\R^n$ or the Laplace-Beltrami operator on compact manifolds. The Witten Laplacian on forms on
compact manifolds is discussed in \cite{helffer-sjostrand-4}.
\end{remark}

\begin{remark} \label{RemarkTransport0}
The coefficients $a_{j,k}$ from Thm.\ \ref{Theorem1} necessarily fulfill the recursive transport equations
\begin{equation} \label{TransportEquations}
  (\nabla_{2\grad\phi}^\Eh + W + \Delta \phi - E_0){a}_{j, k} = -  L {a}_{j,k-1} + \sum\nolimits_{i=1/2}^{k} E_{j, i}\,{a}_{j,k-i}
\end{equation}
where $\nabla^\Eh$ is the connection from Remark \ref{RemarkLaplaceType}. Hence one could try to solve these equations in order to prove 
our theorem. This works indeed well in the non-degenerate case, i.e.\ when the multiplicity of the eigenvalue $\hbar E_0$ of the local harmonic oscillator 
is equal to one, choosing as $a_{j,0}$ the unique function in the kernel of 
$Q= \nabla_{2\grad\phi}^\Eh + W + \Delta \phi - E_0$ (see \cite{dima} or \cite[Ex.\ 8.3]{matthias} for a general treatment of equations of this sort). However, in the degenerate case, i.e. when the multiplicity of the eigenvalue $\hbar E_0$ of 
the local harmonic oscillator is greater that one, it is not even clear how to choose an appropriate element $a_{j,0}$ in $\ker Q$. 
To overcome this difficulty, it is natural to use some kind of spectral projections, given by power series in $\sqrt{h}$. Such an approach was 
used in both \cite{helffer-sjostrand-1} (using an FBI-transform) and \cite{klein-schwarz}, as well as in \cite{klein-rosen1} (using a slightly more elementary
transformation by conjugation). We adapt the latter approach to our case.
\end{remark}

The paper is organized as follows. In Section \ref{Kapitel_2} we introduce the Taylor series map $\tau_p$ on sections and differential operators with respect 
to a fixed normal geodesic chart $x$ at $p\in M$, leading to the space $\Eh_p[[x]]$ of formal power series in $x$ with values in $\Eh_p$. We use the 
stationary phase
method in its real form to define an inner product on $\Ka := \Eh_p[[x]]\lau$, the space of formal Laurent series in the variable $\hbar^{1/2}$ that have 
elements of $\Eh_p[[x]]$ as coefficients. 

In Section \ref{Section4}, we define the rescaling operator $R$, setting $x = \hbar^{1/2}y$, 
which maps $\Ka$ to $\Ka_0$, a subspace of the space $\Eh_p[y]\lau$ of Laurent series in the variable $\hbar^{1/2}$ with coefficients in the polynomial 
space $\Eh_p[y]$.

In Section \ref{Kapitel3} we use Taylor series and the rescaling operator to define the operator $\hbar Q= R \circ\tau_p(H_{\phi, \hbar}) \circ R^{-1}$ on
$\Ka_0$ where $H_{\phi, \hbar} = e^{\phi/\hbar} \circ H_\hbar \circ e^{-\phi/\hbar}$. To a given eigenvalue of the leading order $Q_0$, we then construct 
eigenfunctions and eigenvalues of $Q$ and $\tau_p(H_{\phi, \hbar})$ and prove results on the absence of integer or half-integer order terms in the expansion 
with respect to $\hbar$. 

Finally, the proof of Theorem \ref{Theorem1} is given in Section \ref{Kapitel4}.

\section{Notation and first Constructions}\label{Kapitel_2}
Let $M$ be a manifold and $p \in M$. With respect to a chart $x$ on $M$ with $x(p) = 0$, any function $f\in C^\infty(U)$ has a Taylor series at $p$
\begin{equation}\label{taylorf}
  f \sim \sum\nolimits_{\alpha \in \N_0^n} f_\alpha x^\alpha=: \tau_{p,x}(f)\subset \C[[x]], ~~~~~ f_\alpha \in \C\, ,
\end{equation}
which is determined by the property that for each neighborhood $U^\prime$ of $p$ compactly contained in the domain
of $x$ and for each $N\in\N_0$, there exists a constant $C_{N,U^\prime}>0$ such that
\[ \Bigl|f-\sum\nolimits_{|\alpha|\leq N} f_\alpha x^\alpha \Bigr| \leq C_{N,U^\prime} |x|^{N+1}\quad\text{ uniformly on }\, U\, .\] 
By $\C[[x]]$, we denote the space of formal power series in the $n$ variables $x^1, \ldots, x^n$ and, by
abuse of notation, we identify the coordinate functions of the chart chart $x$ with variables $x^1, \ldots x^n$. We
call $\tau_{p,x}(f)\in \C[[x]]$ defined in \eqref{taylorf} the Taylor series of $f$ at $p$ (with respect to $x$).

If $x$ and $\tilde{x}$ are normal coordinates with respect to the Riemannian metric, then $x=Q\circ\tilde{x}$ for some matrix $Q\in O(n)$.
On the other hand, $Q$ induces an algebra isomorphism 
\[ \tilde{Q}: \C[[x]] \rightarrow \C[[\tilde{x}]] \]
via $\tilde{Q}(x^i) = \sum_j Q^i_j\tilde{x}^j$. Thus $\tau_{p,\tilde{x}} = \tilde{Q}\circ\tau_{p,x}$.

From now on, we work throughout in the Setup \ref{setup1} and we fix an admissible pair $(U, \phi)$, remember Def.~\ref{DefAdmissible}. 

\begin{definition}\label{setup3} 
We choose normal coordinates $x$ such that  for our fixed solution of \eqref{EikonalEquation}, we have
\begin{equation}\label{taupphi2}
\tau_p (\phi) = \sum\nolimits_{j=1}^n \lambda_j (x^j)^2 + O\bigl(|x|^3\bigr)\, , \qquad \lambda_1, \ldots \lambda_n >0\, , 
\end{equation}
near $p$ and write $\tau_p$ instead of $\tau_{p,x}$. We call $\tau_p$ the {\em geodesic Taylor series} map.
\end{definition}

All $\lambda_j\in \R$ in \eqref{taupphi2} are strictly positive since by Setup \ref{setup1}, the minimum of $V$ at $p$ is non-degenerate. 

\begin{remark}\label{CTpM}
The map
\[ x^i \mapsto dx^i|_p\in T^*_pM\cong \{\text{ homogeneous polynomials of degree 1 on } T_pM\,\}  \]
induces an $\C$-algebra-isomorphism $\Phi_x$ from $\C[x]$ to $\C[T_pM]$, the space of complex valued
polynomial functions on $T_pM$. $\Phi_x$ can be extended to the respective completions with respect to the
valuations by polynomial degree, $\C[[x]]$ and $\C[[T_pM]]$. It is easy to see that 
$\Phi_x\circ \tau_{p,x} = \Phi_{\tilde{x}}\circ \tau_{p,\tilde{x}}$ for all normal coordinates $x$ and $\tilde{x}$.
\end{remark}

To a section $u\in \Gamma^\infty (M, \Eh)$ we associate a geodesic Taylor series at $p$ in the following way:
We trivialize $\Eh$ by identifying the fibers along geodesics emanating from $p$ by parallel translation
with respect to the connection $\nabla^\Eh$ given in Remark \ref{RemarkLaplaceType}. Near $p$, $u$ can then be seen as a smooth function with 
values in the
vector space $\Eh_p$ (the fiber of $\Eh$ over the point $p$). In this sense it has a Taylor series $\tau_p(u)$ at $p$ with respect
to the normal coordinates $x$ of Def.\ \ref{setup3}, which will then map to the space $\Eh_p[[x]] := \Eh_p\otimes \C[[x]]$.

Finally, under the trivialization of $\Eh$ above, 
a differential operator $P$ of order $k\in\N$ acting on sections of $\Eh$
can be seen as a differential operator acting on $\Eh_p$-valued functions. By Taylor expansion of its
coefficients, it
 has a geodesic Taylor series $\tau_p(P)$ of the form
\begin{equation}\label{tau_p_von_P}
  P \sim \sum_{\natop{\alpha, \beta\in\N_0^n}{|\beta|\leq k}} P_{\alpha\beta} \, x^\alpha \frac{\partial^{|\beta|}}{\partial x^\beta}=: 
\tau_p(P)\subset \mathscr{D}(\Eh_p)[[x]]
  ~~~~~~~ \text{where} ~~ P_{\alpha\beta} \in \End(\Eh_p)
\end{equation}
and  $\mathscr{D}(\Eh_p)[[x]]$ denotes the space of differential operators with coefficients in $\End(\Eh_p)[[x]]$. These operators act on the space 
$\Eh_p[[x]]$ by formal derivation term by term (inducing a left-module structure on that space). Let us remark that by construction, we have 
\begin{equation} \label{MultiplicationProperty}
  \tau_p(Pu) = \tau_p(P)\tau_p(u).
\end{equation}
In all this, we define $\tau_p$ with respect to the chart from Def.\ \ref{setup3}.

\begin{definition}\label{Def_pgrading}
We say that $\boldsymbol{P}\in \mathscr{D}(\Eh_p)[[x]]$ is homogeneous of degree $\deg_{\mathscr{D}} \boldsymbol{P} = j$, $j\in \Z,$ if and
only if $\boldsymbol{P}$ maps homogeneous polynomials of degree $k$ to homogeneous polynomials of degree $k+j$. In this case, $\boldsymbol{P}$
is a (necessarily finite) sum 
\[ \boldsymbol{P} = \sum_{\natop{\alpha, \beta\in \N_0^n}{|\alpha| - |\beta| = j}} P_{\alpha\beta}x^\alpha \frac{\partial^{|\beta|}}{\partial x^\beta} \; ,\qquad
P_{\alpha\beta} \in \End(\Eh_p)\, . \]
\end{definition}

We recall Borel's Theorem which (for the case $M=\R^n$) is proven  for example in \cite[Thm.\ 1.2.6]{hormander}. 
Since the statement is purely local, it generalizes to our setting and gives

\begin{corollary} \label{BorelTheorem}
The map $\tau_p: \Gamma^\infty (U, \Eh) \rightarrow \Eh_p [[x]]$ is surjective.
\end{corollary}

We define a weighted inner product on the space 
$\Gamma_c^\infty(U, \Eh)$ of sections, compactly supported on $U$, by
\begin{equation}\label{Scalar_phi}
  \bigl( u, w \bigr)_{\gamma,\phi} := \int_U \gamma[u, w] e^{-2\phi/\hbar},
\end{equation}
where $\phi$ is our fixed solution of \eqref{EikonalEquation}. (To be precise, this defines a whole family of inner products, depending on $\hbar \in (0, \infty)$.)

We introduce the space 
\begin{equation} \label{DefinitionKa}
  \Ka := \Eh_p[[x]] \lau
\end{equation}
of
formal Laurent series in $\hbar^{1/2}$ with coefficients in $\Eh_p[[x]]$ which is a vector space over the field 
$\C\lau$ of formal Laurent series in the variable $\hbar^{1/2}$. 
On $\Ka$ we shall
define a non-degenerate Hermitian form $\bigl(\,\cdot\,,\,\cdot\,\bigr)_{\Ka,\phi}$ with values in $\C\lau$,
by using the stationary phase approximation (see e.g.\ \cite{gs}). 

We remark that for each admissible pair $(U, \phi)$, there is a Morse-chart $\kappa$ defined on some neighborhood $U^\prime \subset U$ of $p$ with 
$\phi=|\kappa|^2$. This chart is used in the following theorem.

\begin{theorem}[Method of Stationary Phase]
For $\hbar>0$ and $f \in C^\infty_c(U)$, we set
\begin{equation}\label{Iphi}
 I_\phi(f, \hbar) := \hbar^{-n/2} \int_M f \, e^{-2\phi/\hbar} \, .
\end{equation}
Then, at $\hbar = 0$, the function $I_\phi(f; \hbar)$ has  the asymptotic expansion
\begin{equation} \label{StationaryPhaseExpansion}
  I_\phi (f; \hbar)~~ \sim~~ \left( \pi/2 \right)^{n/2} \sum_{k=0}^\infty \hbar^k \frac{1}{k!} \, 
\Delta^k \Bigl[ (f\circ \kappa^{-1}) \cdot \det\nolimits^{1/2} \bigl(g_{ij}\bigr)\Bigr](0)
\end{equation}
where $\Delta$ is the Laplacian on $\R^n$ and $g_{ij}$ is the matrix of the metric with respect to the Morse-chart $\kappa$. 
\end{theorem}

In particular, we obtain a linear map 
\begin{equation}\label{mathcalI}
\mathcal{I}: C_c^\infty(U) \longrightarrow \C\lau\, , \qquad f \longmapsto \text{right hand side of}~~\eqref{StationaryPhaseExpansion}\, .
\end{equation}

By the above formula, it is clear that the asymptotic expansion of $I_\phi(f, \hbar)$ only depends on the 
Taylor series of $f, \phi$ and $\det\,^{\!\!1/2} \bigl( g_{ij}(\kappa)\bigr)$ at $p\in M$. Thus the kernel of $\tau_p$ is contained in the kernel of 
$\mathcal{I}$ defined in \eqref{mathcalI}. 
Since $\tau_p$ is surjective by Corollary \ref{BorelTheorem}, there is a unique linear map 
\begin{equation}\label{hatI}
\hat{\mathcal{I}}: \C[[x]] \rightarrow \C\lau
\end{equation} 
such that the diagram
\begin{equation} \label{diagram}
\xymatrix{
   C_c^\infty(U) \ar[rr]^{\mathcal{I}} \ar[d]_{\tau_p} &  &\C\lau\\
  \C[[x]] \ar[urr]_{\hat{\mathcal{I}}}
}
\end{equation}
 commutes.
 
We use this construction for $\gamma[u, w]\in C_c^\infty (U)$ where $u, w \in \Gamma_c^\infty(U, \Eh)$. The inner product $\gamma$ on $\Eh$ has a 
Taylor series 
\[ 
\tau_p(\gamma) = \sum_{\delta\in \N_0^n} \gamma_\delta x^\delta \in (\Eh_p^* \otimes \Eh_p^*)[[x]]  \] 
which can be interpreted as a positive definite Hermitian map $\tau_p(\gamma): \Eh_p[[x]] \times \Eh_p[[x]] \longrightarrow \C[[x]]$
by setting
\begin{equation}\label{taupgamma}
  \tau_p(\gamma)\left[ \sum_\alpha u_\alpha x^\alpha, \sum_{\beta} w_\beta x^\beta \right] 
  := \sum_{\alpha,\beta,\delta\in\N_0^n} \gamma_\delta[u_\alpha, w_\beta] x^{\alpha + \beta + \delta}.
\end{equation}
By sesquilinearity, $\tau_p(\gamma)$ extends to a positive definite Hermitian map
\begin{equation}\label{taupgammadef}
\tau_p (\gamma): \Ka \times \Ka \longrightarrow \C[[x]]\lau,
\end{equation}
where $\Ka$ was defined in \eqref{DefinitionKa}. By construction, we have
\begin{equation} \label{CommutativityScalarProduct}
\tau_p(\gamma)[\tau_p(u), \tau_p(w)] = \tau_p(\gamma[u, w]).
\end{equation}

Now we define the hermitan form $(\, \cdot\,,\,\cdot\,)_{\Ka,\phi}$ first on $\Eh_p[[x]]$ by  
\begin{equation}\label{scalarprodKdef}
   \left( \boldsymbol{u}, \boldsymbol{w}  \right)_{\Ka,\phi} 
   := \hat{\mathcal{I}} \bigl( \tau_p(\gamma)\!\left[ \boldsymbol{u}, \boldsymbol{w}\right]\bigr)\, ,\qquad \boldsymbol{u}, \boldsymbol{w} \in \Eh_p[[x]]\, .
\end{equation}
By commutativity of the diagram \eqref{diagram} and \eqref{CommutativityScalarProduct}, it satisfies
\begin{equation}\label{scalarprodmitI}
   \left( \tau_p(u), \tau_p(w)  \right)_{\Ka,\phi} = \mathcal{I}(\gamma[u, w])
\end{equation}
and naturally extends to a Hermitian form
\begin{equation}\label{scalarprodKdef2}
 ( \,\cdot\,, \,\cdot\, )_{\Ka,\phi}: \Ka \times \Ka \longrightarrow \C\lau\, .
\end{equation}
Moreover, by the sesquilinearity of $\gamma$,  $\mathcal{I}(\gamma[\cdot, \cdot])$ extends to a hermitian form 
\begin{equation}\label{extendI}
 \mathcal{I}(\gamma[\cdot, \cdot])\ : \mathcal{F}\times \mathcal{F} \rightarrow \C (\!(\hbar^{1/2})\!) 
\end{equation}
on the space $\mathcal{F}:=\Gamma_c^\infty (U, \Eh) ((\hbar^{1/2}))$
of Laurent series in $\hbar^{1/2}$ with coefficients in $ \Gamma_c^\infty (U, \Eh)$.

\begin{lemma}\label{skpnondegLem}
 The Hermitian form $( \,\cdot\,, \,\cdot\, )_{\Ka,\phi}$ defined in \eqref{scalarprodKdef2} is non-degenerate.
\end{lemma}

\begin{proof}
We have to show that $(\boldsymbol{u}, \boldsymbol{w})_{\Ka,\phi} = 0$ implies $\boldsymbol{u}=0$ or $\boldsymbol{v} = 0$
for all $\boldsymbol{u}, \boldsymbol{w} \in \Ka$. It suffices to show that $( \,\cdot\,, \,\cdot\, )_{\Ka,\phi}$ is non-degenerate on $\Eh_p[[x]]$, so 
we may assume $\boldsymbol{u}, \boldsymbol{w} \in \Eh_p[[x]]$. 

Since the map $\tau_p: \Gamma_c^\infty(U, \Eh) \longrightarrow \Eh_p[[x]]$ is surjective, we may assume that $\boldsymbol{u} = \tau_p(u)$ and 
$\boldsymbol{w} = \tau_p(w)$ for some $u, v \in \Gamma_c^\infty(U, \Eh)$. From the definition it is clear that 
$(\boldsymbol{u}, \boldsymbol{w})_{\Ka,\phi} = 0$ if and only if the function $\gamma[u, w]$ vanishes to infinite order at the point $p$. In this case, however, 
we necessarily have that either $u$ or $w$ vanish to infinite order at $p$. Therefore $\tau_p(u) = \boldsymbol{u}= 0$ or $\tau_p(w) = \boldsymbol{w} = 0$.
\end{proof}

The reason to define $( \,\cdot\,, \,\cdot\, )_{\Ka,\phi}$ as we did is that it has the following property.

\begin{proposition}[Symmetry] \label{SymmetrieOperator}
  Let $P$ be a differential operator, acting on sections of $\Eh$, which is symmetric on $\Gamma_c^\infty(U, \Eh)$ with respect to the inner product $\bigl(\,.\,,\,.\,\bigr)_{\gamma}$ defined in 
  \eqref{standard_skalar}. For our fixed solution $\phi$ of \eqref{EikonalEquation}, we define the conjugated operator $P_\phi$ on $\Gamma^\infty (U, \Eh)$ 
with respect to $e^{-\phi/\hbar}$ by 
\begin{equation} \label{ConjugationByPhi}
  P_\phi u := e^{\phi/\hbar} P \bigl[e^{-\phi/\hbar}u\bigr]\, , \qquad u\in \Gamma^\infty (U, \Eh)\; .
\end{equation}  
 Then $P_\phi$ is symmetric on $\Gamma_c^\infty(U, \Eh)$ with respect to $(\,.\,,\,.\,)_{\gamma, \phi}$ (defined above) and its 
Taylor series $\tau_p(P_\phi)$ (see \eqref{tau_p_von_P}) 
acts on $\Ka$ and is symmetric with respect to the sesquilinear form $( \,\cdot\,, \,\cdot\, )_{\Ka, \phi}$ defined in \eqref{scalarprodKdef2}.
\end{proposition}

\begin{proof}
It is clear that $\tau_p(P_\phi)\in \mathscr{D}(\Eh_p)[[x]]$ extends to a differential operator on $\Ka$.
By sesquilinearity, it suffices to prove the statement on the symmetry for elements of $\Ka$ which are constant in $\hbar^{1/2}$. 
Since by Corollary \ref{BorelTheorem} all elements of $\Eh_p[[x]]$ can be written as Taylor series $\tau_p(u)$ for some $u \in \Gamma_c^\infty(U, \Eh)$, 
we need to prove the statement only for elements of this type. 
Now by \eqref{MultiplicationProperty} and \eqref{scalarprodmitI}
\begin{align}
  \bigl( \tau_p(P)\tau_p(u), \tau_p(w) \bigr)_{\Ka,\phi} &=\bigl( \tau_p(Pu), \tau_p(w) \bigr)_{\Ka,\phi}\nonumber\\ 
  &= \mathcal{I}\bigl(\gamma[Pu, w]\bigr) \label{lemma4.5.1} \\
  &= \text{asymptotic expansion of}~ I_\phi\bigl(\gamma[Pu, w], \hbar\bigr)\nonumber\; .
\end{align}
By \eqref{Iphi}, \eqref{Scalar_phi}, \eqref{standard_skalar} and the symmetry of $P$ with respect to $\bigl(\,.\,,\,.\,\bigr)_{\gamma}$ we get for 
$u,w\in\Gamma_c^\infty(U, \Eh)$
\begin{align*}
I_\phi\bigl(\gamma[Pu, w], \hbar\bigr) &= \hbar^{-n/2} \bigl(P_\phi u, w\bigr)_{\gamma,\phi} = 
\hbar^{-n/2} \bigl(P e^{-\phi/\hbar}u, e^{-\phi/\hbar}w\bigr)_{\gamma} \\
&= \hbar^{-n/2} \bigl(e^{-\phi/\hbar}u, P e^{-\phi/\hbar}w\bigr)_{\gamma} = \hbar^{-n/2} 
\bigl(u, P_\phi w\bigr)_{\gamma,\phi} = I_\phi\bigl(\gamma[u, Pw], \hbar\bigr)
\end{align*}
and using \eqref{lemma4.5.1} again gives the stated result.
\end{proof}

\section{Rescaling} \label{Section4}

In order to analyze degenerate eigenvalues, we now consider the rescaled variable $y=\hbar^{-1/2}x$ instead of $x$, where $x$ are the coordinates from 
Def.\ \eqref{setup3}. This is motivated by 
the basic scaling property of $H_{p,\hbar}$, giving the scaling in $\spec (H_{p,\hbar})$ and of the associated eigenfunctions. 

The important observation here is the following: Given a formal power series in $\hbar$ with coefficients in $\Gamma^\infty(U, \Eh)$ as in Thm.\ 
\ref{Theorem1}, its Taylor series will be an element of $\Ka$, i.e., a formal power series in $\hbar^{1/2}$ whose coefficients are formal power series in the variable
 $x$. After rescaling, however, it becomes an element of a space $\Ka_0$ (defined in \eqref{DefinitionKa0} below), which is the space of formal power series 
in $\hbar$ whose coefficients are {\em polynomials} in the rescaled variable $y$.

In this section, we explain the details. Again, throughout this section, we work in Setup \ref{setup1} and fix an admissible pair $(U, \phi)$.

\begin{definition}[Rescaling Operator]\label{rescaleDef}
We define the {\em rescaling operator} $R: \Ka \longrightarrow \Eh_p[y] \lau$ by 
\begin{equation*}
  \boldsymbol{u} = \sum_{\natop{k \in \Z/2}{k\geq -M}} \hbar^k \sum_{\alpha \in \N_0^n} u_{\alpha, k} x^\alpha\quad \longmapsto \quad
  R\boldsymbol{u} = \sum_{\natop{k \in \Z/2}{k\geq -M}}\sum_{\alpha\in \N_0^n}\hbar^{k + |\alpha|/2} u_{\alpha,k} y^\alpha  \; .
\end{equation*}
\end{definition}

\begin{proposition}[Image of $R$]\label{rescaleProp}
The rescaling operator $R$ given in Definition \ref{rescaleDef}
is well-defined, injective and its image is the space
\begin{equation} \label{DefinitionKa0}
\Ka_0 := \Bigl\{ \hbar^{-K}\!\!\!\sum_{j \in \N_0/2} P_j(y) \, \hbar^j \in \Eh_p[y] \lau \mid K \in \N_0/2 ,~ \mathrm{deg} \,P_j(y) \leq 2j \Bigr\}
\end{equation}
\end{proposition}

\begin{proof}
An element $\boldsymbol{u} \in \Ka$ can be written as
\begin{equation} \label{formofu}
  \boldsymbol{u} = \hbar^{-K}\!\!\!\sum_{k \in \N_0/2} \hbar^k \sum_{\alpha \in \N_0^n} u_{\alpha, k} x^\alpha
\end{equation}
for some $K\in\N_0$. Hence
\begin{equation}\label{Rhatu}
 R \boldsymbol{u} 
 = \hbar^{-K}\!\!\!\sum_{k \in \N_0/2} \hbar^k \sum_{\alpha \in \N_0^n} u_{\alpha, k} \,y^\alpha \,\hbar^{|\alpha|/2}
 = \hbar^{-K}\!\!\!\sum_{k \in \N_0/2} \hbar^k \sum_{j=0}^{2k} \sum_{|\alpha|=j} u_{\alpha, k-j/2} \,y^\alpha,
\end{equation}
so $R \boldsymbol{u}\in\Eh_p[y]\lau$ and if $R\boldsymbol{u} = 0$, then each $u_{\alpha, k}$ has to be zero, i.e.\ $\boldsymbol{u} = 0$. 
This shows that $R$ is well-defined and injective. Furthermore, it is clear 
from \eqref{Rhatu} that $R\Ka\subset \Ka_0$. On the other hand, given an element
\begin{equation*}
  \boldsymbol{w} = \hbar^{-K}\!\!\!\sum_{j \in \N_0/2} \hbar^k \sum_{|\alpha|\leq 2k} w_{\alpha, k} \,x^{\alpha}\in \Ka_0\, ,
\end{equation*}
its preimage $\boldsymbol{u}\in \Ka$ under $R$ in the form \eqref{formofu} has the coefficients $u_{\alpha, k} := w_{\alpha, k + |\alpha|/2}$, hence the image of $R$ is all of $\Ka_0$.
\end{proof}

\begin{definition}\label{DefProdK_0}
By Prop.\ \ref{rescaleProp}, we can define the Hermitian form 
\begin{equation}\label{binlinK0}
 ( \,\cdot\,, \,\cdot\,)_{\Ka_0,\phi}:= (R^{-1})^*\bigl[( \,\cdot\,, \,\cdot\,)_{\Ka,\phi}\bigr] : \Ka_0\times \Ka_0\rightarrow \C\lau
\end{equation}
as the pullback of the inner product $( \,\cdot\,, \,\cdot\,)_{\Ka,\phi}$ on $\Ka$ defined in \eqref{scalarprodKdef} under the inverse
rescaling operator $R^{-1}: \Ka_0 \rightarrow \Ka$.
\end{definition}

Analogously to Def.\ \ref{rescaleDef}, we define (with $x={\hbar}^{1/2}y$) rescaling operators from $(\Eh^*_p\otimes \Eh^*_p)[[x]]\lau$ to 
$(\Eh^*_p\otimes \Eh^*_p)[y]\lau$ and from $\C[[x]]\lau$ to $\C[y]\lau$ which we also denote by $R$.
Then, for $\tau_p(\gamma)\in (\Eh^*_p\otimes \Eh^*_p)[[x]]$ given in \eqref{taupgamma}, \eqref{taupgammadef}, we get
\begin{equation}\label{gammak}
R\circ \tau_p(\gamma) = \sum_{\delta\in \N_0^n} \hbar^{|\delta|/2} \gamma_\delta y^\delta =: 
\sum_{k\in\N_0/2} \hbar^k \gamma_k \quad\text{with}\quad 
\gamma_k := \sum_{|\delta|= 2k} \gamma_\delta y^\delta\in (\Eh^*_p\otimes \Eh^*_p)[y] \, .
\end{equation}
Remember that $\tau_p$ always refers to the coordinates $x$ as defined in Def.\ \ref{setup3}.
We also introduce the positive definite Hermitian form $\boldsymbol{\gamma}: \Ka_0\times \Ka_0 \rightarrow \C[y]\lau$, by setting
for $\boldsymbol{u}, \boldsymbol{v}\in\Ka_0$
\begin{align}\label{tildegammaDef1}
\boldsymbol{\gamma}[\boldsymbol{u}, \boldsymbol{v}] &:= \hbar^{-K_1-K_2}\sum_{j\in \N_0/2} \hbar^j \sum_{k+\ell+ r = j} 
\gamma_r[u_k, v_\ell] 
\end{align}
where
\begin{align}
\gamma_r[u_k, v_\ell] &:= \sum_{\natop{\alpha, \beta, \delta \in\N_0^n}{|\alpha|\leq 2k, |\beta|\leq 2\ell, |\delta| = 2r}} 
\gamma_\delta [u_{k,\alpha}, v_{\ell, \beta}] y^{\alpha + \beta + \delta}\in \C[y]\quad \text{for } \gamma_\delta \text{ as in }
\eqref{gammak}\, ,\label{tildegammaDef2}\\
\boldsymbol{u}&=\hbar^{-K_1}\sum_{k\in\N_0/2}\hbar^k u_k\quad \text{with}\quad u_k = \sum_{|\alpha|\leq 2k} u_{k, \alpha} y^\alpha \quad
\text{and}\label{tildegammaDef3}\\
\boldsymbol{v}&=\hbar^{-K_2}\sum_{\ell\in\N_0/2}\hbar^\ell v_\ell\quad \text{with}\quad v_\ell = \sum_{|\beta|\leq 2\ell} v_{\ell, \beta} y^\beta\; .
\label{tildegammaDef4}
\end{align}
Then for all $\boldsymbol{u}, \boldsymbol{v}\in\Ka$, we have
\begin{equation}\label{Rundgammatilde}
\boldsymbol{\gamma}[R\boldsymbol{u}, R\boldsymbol{v}] = R\bigl(\tau_p(\gamma)[\boldsymbol{u}, \boldsymbol{v}]\bigr)\quad\text{and}\quad \mathrm{deg}\,\gamma_r[u_k, v_\ell]\leq 2(r+k+\ell).
\end{equation}
We say that $u\in \Eh_p[y]$ (or $\C[y]$) has parity $\pm 1$, if and only if $u(y) = \pm u(-y)$, i.e.\ $u$ only contains monomials of even (for $+$) or of odd degree (for $-$).
From \eqref{tildegammaDef2} it follows that if $u_k$ has the parity $(-1)^{2k}$ and $v_\ell$ has the parity $(-1)^{2\ell}$, then 
$\gamma_r[u_k, v_\ell]$ has the parity
$(-1)^{2k + 2\ell + 2r}$.

\begin{proposition} \label{rescaleProp2}
The Hermitian form $( \,\cdot\,, \,\cdot\,)_{\Ka_0,\phi}$ defined in \eqref{binlinK0} is 
non-degenerate. An 
operator $P$ on $\Ka$ is 
symmetric with respect to $( \,\cdot\,, \,\cdot\,)_{\Ka, \phi}$ if and only if the associated operator 
$R\circ P\circ R^{-1}$ on $\Ka_0$ is
symmetric with respect to $( \,\cdot\,, \,\cdot\,)_{\Ka_0, \phi}$.\\
Moreover, there exist polynomials 
$\omega_k\in\Eh_p [y],\, k\in \N_0/2$
of order $2k$ and parity $(-1)^{2k}$ such that 
for $\boldsymbol{u}, \boldsymbol{v}\in \Ka_0$ and $\gamma_r[u_k, v_\ell]\in \C[y]$ as given in \eqref{tildegammaDef2},
\eqref{tildegammaDef3} and \eqref{tildegammaDef4}
\begin{equation}\label{skpk}
\bigl(\boldsymbol{u}, \boldsymbol{v} \bigr)_{\Ka_0, \phi} = 
\hbar^{-K_1-K_2}\sum_{k \in \N_0/2} \hbar^k \sum_{\natop{j,\ell, r, m\in\N_0/2}{j + \ell + r + m = k}} 
\int_{\R^n} \gamma_r [u_j, v_\ell] (y)  \omega_m(y)
e^{-\langle y, \Lambda y\rangle}  \, dy
\end{equation}
where $\Lambda = \diag (\lambda_1, \ldots, \lambda_n)\in \mathrm{Mat}(n\times n, \R)$ for $\lambda_\nu>0$ as given
in \eqref{taupphi2}.
\end{proposition}

\begin{proof}
By Prop.\ \ref{rescaleProp} and Lemma \ref{skpnondegLem}, the Hermitian form $\bigl( \,\cdot\,, \,\cdot\,\bigr)_{\Ka_0, \phi}$ is well-defined 
and non-degenerate. By its definition in \eqref{binlinK0}, we have 
$\bigl(R\circ P\circ R^{-1} \boldsymbol{u}, \boldsymbol{v} \bigr)_{\Ka_0, \phi} = \bigl(P u, v \bigr)_{\Ka, \phi}$ for 
$\boldsymbol{u} = R \boldsymbol{u}, \boldsymbol{v}= R\boldsymbol{v} \in \Ka_0$, proving the stated equivalence of symmetry. 
Using \eqref{binlinK0}, \eqref{scalarprodmitI} and \eqref{mathcalI}
we get for $u, v\in \Gamma_c^\infty (U, \Eh)$
\begin{equation}\label{rescalePropBew1} 
\bigl( R\circ \tau_p (u), R\circ \tau_p (v))_{\Ka_0,\phi} = \mathcal{I}(\gamma[u, v]) = \;\text{asymptotic expansion of }\, I_\phi(\gamma[u,v], \hbar)\, .
\end{equation}
Now $\boldsymbol{u} = R\circ \tau_p(u)$ and $\boldsymbol{v} = R\circ \tau_p(v)$ for some $u, v\in \Gamma_c^\infty (U, \Eh)$ by 
Prop.\ \ref{rescaleProp} and Corollary \ref{BorelTheorem}. 
Identifying $\phi$ and $\gamma[u,v]$ with their pullback under the inverse chart $x^{-1}$, we have
\[ I_\phi (\gamma[u,v], \hbar) = \hbar^{-n/2} \int_{\R^n} \gamma[u,v](x) e^{-\phi(x)/\hbar} G(x)\, dx\; , \]
where  $G(x):= (\det g_{ij}(x))^{1/2}$, using the compact support of $\gamma[u,v]$ in $U$.\\
Changing variables by $x=\hbar^{1/2}y$ gives
\begin{equation}\label{Iphimity}
 I_\phi (\gamma[u,v], \hbar) = \int_{\R^n} \gamma[u,v](\hbar^{1/2} y) e^{-2\phi(\hbar^{1/2}y) /\hbar} G(\hbar^{1/2}y)\, dy\, .
\end{equation}
It is now straightforward (though tedious) to check that the asymptotic expansion on the right hand side of \eqref{Iphimity} is given by the asymptotic expansion of
all three factors in the integrand of \eqref{Iphimity} and integrating term by term (see e.g.\ \cite[Lemma 3.7]{klein-rosen1} for a similar discussion). 

Remember that we chose coordinates such that the Hessian of $\phi$ is diagonal (see Def.\ \ref{setup3}).

It is straightforward to check from the definitions that
\begin{equation}\label{asym_gamma}
 \text{asymptotic expansion of }\gamma[u, v] (\hbar^{1/2}y) = \boldsymbol{\gamma} [\boldsymbol{u}, \boldsymbol{v}] \in \C[y]\lau\; .
\end{equation}
Furthermore, by \eqref{taupphi2} the asymptotic expansion of the phase-function $\phi\in C^\infty(U, \R)$ is given by
\begin{equation}\label{sigmapphi} 
\text{asymptotic expansion of }\frac{1}{\hbar} \phi (\hbar^{1/2}y) = \frac{1}{2}\<y, \Lambda y\> + 
\sum_{k\in\N} \hbar^{k/2} \phi_k 
\end{equation}
where $\Lambda = \diag (\lambda_1, \ldots \lambda_n)$ is strictly positive and $\phi_k\in \C[y]$ are homogeneous of degree $k+2$. 
Equation \eqref{sigmapphi} together with \eqref{gij} allows us to define real polynomials $\omega_k\in\R[y]$ by  
\begin{equation}\label{eminus2varphi}
\text{asymptotic expansion of } e^{-2\phi(\hbar^{1/2}y)/\hbar} G(\hbar^{1/2}y) =:
e^{-\langle y, \Lambda y\rangle}
\sum_{k\in\N_0/2}\hbar^k \omega_k \in \R[y]\lau 
\end{equation}
where $\omega_0:= 1$ and the parity of $\omega_k$ is $(-1)^{2k}$ (this point follows as in the similar discussion in \cite[Rem.\ 1]{klein-schwarz} and the Cauchy-product).
Inserting the expansions \eqref{asym_gamma}, \eqref{eminus2varphi} and \eqref{tildegammaDef1} into \eqref{Iphimity} and ordering by powers of $\hbar$
gives
\begin{equation}
 \text{asymptotic expansion of } I_\phi (\gamma[u,v]; \hbar) = \text{right hand side of }~\eqref{skpk}\; .
\end{equation}
Using \eqref{rescalePropBew1} and the uniqueness of asymptotic expansion proves \eqref{skpk}.
\end{proof}

\begin{remark}
Obviously \eqref{skpk} could be used to define the inner product $( \,\cdot\,, \,\cdot\,)_{\Ka_0, \phi}$ on $\Ka_0$ directly.
This approach avoids the stationary phase expansion \eqref{StationaryPhaseExpansion} used in the definition \eqref{scalarprodKdef}
of the inner product on $\Ka$. But then it is slightly more complicated to show symmetry of 
$Q:=  \hbar^{-1}\bigl(R \circ \tau_p(H_\phi) \circ R^{-1}\bigr)$ given in \eqref{RescaledSeriesOfH} with respect
to $( \,\cdot\,, \,\cdot\,)_{\Ka_0, \phi}$ (see \cite{klein-schwarz}).
\end{remark}

\section{The formal Eigenvalue Problem}\label{Kapitel3}

Again we assume Setup \ref{setup1} and fix an admissible pair $(U, \phi)$.

The operator $\tau_p(H_{\phi,\hbar})$ is symmetric on $\Ka$ by Thm.\ \ref{SymmetrieOperator}. 
In this section, we will consider the formal eigenvalue problem of this operator over the field $\C\lau$, i.e.\ 
we will solve the eigenvalue equation
\begin{equation} \label{SpectralDecomposition}
  \tau_p(H_{\phi,\hbar}) \, \boldsymbol{a} = \boldsymbol{E} \boldsymbol{a},  ~~~~ \text{for} ~~\boldsymbol{a} \in \Ka ~~\text{and}~~
\boldsymbol{E} \in \C\lau.
\end{equation}
We start giving an explicit calculation of the rescaled Taylor series of $H_{\phi,\hbar}$,  as defined in \eqref{ConjugationByPhi}. 

\begin{lemma}\label{Hphi}
On $U$, the conjugation of $H_\hbar$ with respect to $e^{-\phi/\hbar}$ as defined in Proposition \ref{SymmetrieOperator} is given by
\begin{equation} \label{LocalFormHPhi}
  H_{\phi,\hbar} = \hbar^2 L + \hbar \bigl( \nabla^\Eh_{2\grad\phi} + W + \Delta \phi \bigr)
\end{equation}
where $\nabla^\Eh$ is the unique metric connection determined by $L$ as described in Remark \ref{RemarkLaplaceType} and $\Delta$ denotes the Laplace-Beltrami operator acting on functions.
\end{lemma}

\begin{proof}
  For any $f \in C^\infty(M)$ and $u \in \Gamma^\infty(M, \Eh)$, we have the product rule 
  \begin{equation*}
   L(fu) = (\Delta f) u - 2 \nabla^\Eh_{\grad f} u + f Lu,
  \end{equation*}
compare \cite[Prop.\ 2.5]{bgv}. Straightforward calculation gives
\begin{equation} \label{ConjugationFormula}
  H_{\phi,\hbar} = \hbar^2 L + 2 \hbar \nabla^\Eh_{\grad \phi} + \hbar W + \bigl(\hbar \Delta \phi +V - |\dd \phi|^2\bigr)\mathrm{id}_\Eh\, .
\end{equation}
Since $\phi\in C^\infty (U, \R)$ solves the eikonal equation on $U$,  this simplifies to \eqref{LocalFormHPhi}. 
\end{proof}

\begin{remark} \label{RmkTransportEquations}
When making the ansatz 
\begin{equation*}
  e^{-\phi/\hbar}\sum_{k\in \N_0/2}\hbar^k a_k ~~~~ \text{and} ~~~~ \hbar\sum_{k\in \N_0/2}\hbar^k E_k
\end{equation*}
with $a_k \in \Gamma^\infty(U, \Eh)$ for an eigenfunction and eigenvalue of $H_{\phi, \hbar}$, respectively, straightforward calculation gives that the equation
\begin{equation}
   \Bigl(H_{\phi, \hbar} - \hbar\sum\nolimits_{k \in \N_0/2} \hbar^k E_k \Bigr) \sum\nolimits_{k\in \N_0/2}\hbar^k a_k = 0
\end{equation}
in the sense of an asymptotic expansion in $\hbar$ is equivalent to the the statement that for each $k \in \N_0/2$, the coefficients $a_k$ solve the recursive transport equations
\begin{equation} \label{TransportEquation2}
    (\nabla_{2\grad\phi}^\Eh +  W + \Delta \phi - E_0){a}_{k} = -  L {a}_{k-1} + \sum_{i=1/2}^k E_{ i}\,a_{k-i},
\end{equation}
which were discussed in Remark \eqref{RemarkTransport0}. 
\end{remark}

With the chart of Def.\ \ref{setup3}, we have
\begin{equation}\label{taupphi}
\tau_p (\phi) = \frac{1}{2}\langle x, \Lambda x \rangle +   \sum_{k\in\N} \hbar^{k/2} \phi_k
~~~~ \text{and hence} ~~~~
  \tau_p (V) = \< \Lambda x, \Lambda x \> +  \sum_{k\in\N}V_k 
\end{equation}
by \eqref{EikonalEquation} where $\Lambda=\diag (\lambda_1, \ldots, \lambda_n)$ is positive definite and $V_k, \phi_k$ are homogeneous polynomials of degree $k+2$. In the basis $\partial_{x^1}, \dots, \partial_{x^n}$ of $T_pM$, the Hessian $\nabla^2 V|_p$ is given by the matrix $\Lambda^2$, hence the local 
harmonic oscillator at $p$ associated to $H_\hbar$ (see Def. \ref{DefLocalHarmonicOscillator}) is given by
\begin{equation} \label{LocalHarmonicOscillator}
  H_{p, \hbar} = - \hbar^2 \sum_{\nu=1}^n \frac{\partial^2}{\partial X_\nu^2} + \hbar\, W(p) + \< \Lambda X, \Lambda X \>
\end{equation}
at the point $X \in T_pM$.

We use the standard fact of Riemannian geometry that the geodesic Taylor series of the metric is given
by (see e.g.\ \cite[Prop. 1.28]{bgv})
\begin{equation}\label{gij}
\tau_p(g_{ij})=\delta_{ij} - \sum_{k,l}\frac{1}{3}R_{ikjl}x^kx^l + 
\sum_{|\alpha|\geq 3}\frac{\partial g_{ij}}{\partial x^\alpha} \frac{x^\alpha}{\alpha !}\, ,
\end{equation}
where $R_{ikjl}$ denotes the Riemannian curvature tensor.
From \eqref{taupphi} and \eqref{gij} it follows that
\begin{equation}\label{taup_nablagradphi}
  \tau_p(\nabla^\Eh_{2\grad \phi}) =  \sum_{\nu=1}^n 2 \lambda_\nu x^\nu \frac{\partial}{\partial x^\nu}\, \mathrm{id}_\Eh + 
\text{terms of higher degree}
\end{equation}
with respect to the degree from Def.\ \ref{Def_pgrading}. Again by \eqref{gij} we have
\begin{equation}\label{Deltaphi}
  \tau_p(\Delta \phi) = \mathrm{tr}\,\Lambda\, \mathrm{id}_\Eh + \text{terms of higher degree}.
\end{equation} 
Hence, by \eqref{LocalFormHPhi}, \eqref{taup_nablagradphi}, \eqref{Deltaphi} and Def.\ \ref{rescaleDef} the rescaled Taylor series of 
$H_{\phi, \hbar}$ is given by
\begin{equation} \label{RescaledSeriesOfH}
Q := \hbar^{-1}\bigl(R \circ \tau_p(H_{\phi, \hbar}) \circ R^{-1}\bigr) =  \sum_{j \in \N_0/2} \hbar^j Q_j
\end{equation}
for differential operators $Q_j\in \mathscr{D}(\Eh_p)[y]$ independent of $\hbar$. Here, $\mathscr{D}(\Eh_p)[y]$ 
denotes the space of differential operators on $\Eh_p[y]$ with coefficients in $\End (\Eh_p)[y]$ which extend to operators on $\Ka_0$. 
We have in particular 
\begin{equation} \label{defQ0}
  Q_0 = -\sum_{\nu=1}^n  \Bigl( \frac{\partial^2}{\partial y_\nu^2} + 2 \lambda_\nu y^\nu \frac{\partial}{\partial y_\nu} + 
\lambda_\nu \Bigr)\, \mathrm{id}_\Eh + W(p)\; .
\end{equation}

It is clear by \eqref{RescaledSeriesOfH} that $Q$ is a well-defined operator acting on the space $\Ka_0$. 
Moreover, by Thm.\ \ref{SymmetrieOperator} and Prop.\ \ref{rescaleProp2}, $Q$ is symmetric with respect to 
$(\,\cdot\, ,\,\cdot\, )_{\Ka_0, \phi}$. 

\begin{remark} \label{RemarkOnH0AndQ0}
$Q_0$ can be restricted to $\Eh_p[y]$. Using the isomorphism $\Phi_x: \Eh_p[x] \rightarrow \Eh_p [T_pM]$ extending $\Phi_x$ 
defined in Remark \ref{CTpM},
we set $\tilde{R}:= R\circ \Phi_x^{-1} : \Eh_p [T_pM] \rightarrow \Eh_p[y]$.
Then, on $\Eh_p[y]$, $Q_0$ is the rescaling of the conjugation of the local harmonic oscillator $H_{p,\hbar}$ given in 
\eqref{LocalHarmonicOscillator} using the function $\phi_0: X \mapsto \frac{1}{2} \< X, \Lambda X \>$ on $T_pM$, more precisely
\begin{equation*}
  \hbar \tilde{R}^{-1}Q_0 \tilde{R} = e^{\phi_0/\hbar} H_{p,\hbar} e^{-\phi_0/\hbar}.
\end{equation*}
In particular, if a polynomial $q\in \Eh_p[y]$ is an eigenfunction of $Q_0$ with eigenvalue $E_0$, then 
\begin{equation*}
  e^{-\phi_0/\hbar} \tilde{R}^{-1}q (\,\cdot\,) = e^{-\phi_0/\hbar} q(\hbar^{-1/2} \,\cdot\,) \in C^\infty(T_pM, \Eh_p)
\end{equation*}
is an eigenfunction of $H_{p, \hbar}$ with eigenvalue $\hbar E_0$.
\end{remark}

\begin{lemma}\label{Qj}
The operators $Q_j\in \D(\Eh_p)[y],\, j\in \N_0/2,$ in \eqref{RescaledSeriesOfH} are sums $Q_j = L_{2j-2} + A_{2j}$ where $L_k$ and 
$A_k$ are homogeneous with $\mathrm{deg}_\D L_k= k$ and $\mathrm{deg}_\D A_k = k$. 
\end{lemma}

\begin{proof}
If $\tilde{P}\in \mathscr{D}(\Eh_p)[[x]]$ is homogeneous of degree $\deg_\D \tilde{P} = k\in \Z$ and does not depend on $\hbar$, there
are $P_{\alpha\beta}\in\End (\Eh_p)$ such that 
\begin{equation}\label{degreeP} 
R\circ \tilde{P}\circ R^{-1} = \sum_{|\alpha| - |\beta| =k} P_{\alpha\beta}\hbar^{|\alpha|/2}y^\alpha \hbar^{-|\beta|/2}\partial_y^{\beta} =: 
\hbar^{k/2} P
\end{equation}
where $P\in \D(\Eh_p)[y]$ is homogeneous with $\deg_\D P = k$.

In our case, $\tau_p(L) = \sum_{k\in \Z, k\geq -2} \tilde{L}_k$ where $\tilde{L}_k\in \D(\Eh_p)[[x]]$ is homogeneous with $\deg_\D \tilde{L}_k = k$. 
Moreover, by \eqref{taup_nablagradphi} and \eqref{Deltaphi},
$\tau_p\bigl(\nabla_{2\grad \phi} + \Delta \phi + W(p)\bigr) = \sum_{k\in\N_0} \tilde{A}_{k}$ where 
$\tilde{A}_k\in \D(\Eh_p)[[x]]$ is homogeneous with $\deg_\D \tilde{A}_k = k$. Therefore, by \eqref{degreeP}
\begin{align*}
\hbar^{-1} R \circ \hbar^2\tau_p(L) \circ R^{-1} &= \sum_{k\in \Z, k\geq -2} \hbar^{(k+2)/2} L_k  = \sum_{j \in \N_0/2}\hbar^j L_{2j-2} \\
\hbar^{-1}  R \circ \hbar \tau_p\bigl(\nabla_{2\grad \phi} + \Delta \phi + W\bigr) \circ R^{-1} &= \sum_{k\in \N_0} \hbar^{k/2} A_k = 
\sum_{j \in \N_0/2} \hbar^j A_{2j}
\end{align*}
where $\deg_\D L_k = k$ and $\deg_\D A_k = k$. Thus $Q_j = L_{2j-2} + A_{2j}$ and the lemma follows.
\end{proof}

Let $(e_1, \dots, e_{\mathrm{rk}\,\Eh})$ be an orthonormal basis of $\Eh_p$ consisting of eigenvectors of $W(p)$ and let 
$\mu_1, \dots, \mu_{\mathrm{rk}\,\Eh}$ denote the corresponding eigenvalues. Then the eigenvalue problem
\begin{equation} \label{EigenValueEquationQ0}
  Q_0 h(y) = E h(y), ~~~~~~~~ h(y) \in \Eh_p[y], ~~ E \in \R,
\end{equation}
is solved by the eigenvalues
\begin{equation} \label{EigenvaluesHarmonicOscillator}
  E_{\alpha, k} = \sum_{j=1}^n(2 \alpha_j + 1) \lambda_j + \mu_k, ~~~~~~ \alpha \in \N_0^n,~~ 1 \leq k \leq \mathrm{rk}\,\Eh,
\end{equation}
with the corresponding eigenfunctions
\begin{equation} \label{EigenvectorsHarmonicOscillators}
h_{\alpha, k}(y) = \prod_{j=1}^n \sqrt[4]{\lambda_j} ~ h_{\alpha_j}\bigl(\sqrt{\lambda_j}y\bigr) \cdot e_k, ~~~~~~ 
\alpha \in \N_0^n,~~ 1 \leq k \leq \mathrm{rk}\,\Eh,
\end{equation}
where 
\begin{equation}
  h_j(z) = \frac{(-1)^j}{\sqrt{2^j j!}\sqrt[4]{\pi}} e^{z^2}\frac{\dd^j}{\dd z^j} e^{-z^2}, ~~~~~~~
 j \in \N_0,
\end{equation}
are the Hermite polynomials. 

\begin{remark}\label{Remhalphak}
The eigenfunctions $h_{\alpha,k}$ of $Q_0$ are orthonormal with respect to 
$(\,\cdot\, ,\,\cdot\, )_{\Ka_0,\phi}$.
Moreover, since $h_k(-x) = (-1)^k h_k(x)$ and $\deg h_k = k$ it follows that $h_{\alpha,k}$ is even or odd if $|\alpha|$ is even or odd 
respectively and $\deg h_{\alpha,k} = |\alpha|$.
\end{remark}

\begin{remark}\label{SpektrumQ0}
 $Q_0$ can be considered as an unbounded operator on the Hilbert space 
\begin{equation*}
\mathcal{H} := L^2(\R^n, e^{-2\phi_0/\hbar}\, dy)\otimes \Eh_p
\end{equation*}
that is essentially self-adjoint on polynomials with coefficients in $\Eh_p$. By diagonalizing $W(p)$, it is easy to see that the spectrum of 
$Q_0$ on $\mathcal{H}$ is given by
\begin{equation}
  \spec(Q_0) := \{ E_{\alpha, k} \mid \alpha \in \N_0^n, ~~ 1 \leq k \leq \mathrm{rk}\,\Eh \}\; ,
\end{equation}
with $E_{\alpha,k}$ given in \eqref{EigenvaluesHarmonicOscillator}. Moreover, the eigenfunctions $h_{\alpha,k}$ given in 
\eqref{EigenvectorsHarmonicOscillators} are an orthonormal basis of $\mathcal{H}$. In particular, this shows that $Q_0-\mathbf{z}$
is bijective on $\Eh_p[y]$ for $\mathbf{z}\in \C\setminus \spec (Q_0)$. 
\end{remark}

\begin{lemma}[The resolvent]
For $\mathbf{z} \in \C\setminus \spec(Q_0)$, the operator $Q - \mathbf{z}: \Ka_0 \longrightarrow \Ka_0$ is invertible, and the inverse 
$R(\mathbf{z})$ is given 
by the formal Neumann series
\begin{equation} \label{DefinitionOfResolvent}
  R(\mathbf{z}) = \sum_{k=0}^\infty \Bigl( - R_0(\mathbf{z}) \sum_{j \in \N/2} \hbar^j Q_j \Bigr)^k R_0(\mathbf{z}) = 
- \sum_{j \in \N_0/2} \hbar^j R_j(\mathbf{z})
\end{equation}
where $Q_j,\, j\in \N_0/2,$ are the differential operators given in \eqref{RescaledSeriesOfH}
and 
\begin{align}\label{R0}
R_0(\mathbf{z}) &:= (Q_0 - \mathbf{z})^{-1} \\
R_j (\mathbf{z})&:= \sum_{k=1}^{2j} (-1)^k \sum_{\natop{j= j_1 +\ldots + j_k}{j_1, \ldots
j_k \in\frac{\N}{2}}}
\left(\prod_{m=1}^k -R_0(\mathbf{z}) Q_{j_m}\right) R_0(\mathbf{z})\; .\label{Rj}
\end{align}
Furthermore, for all $j \in\N_0/2$
\begin{equation} \label{SymmetryOfRj}
\bigl( \boldsymbol{u}, R(\mathbf{z}) \boldsymbol{w} \bigr)_{\Ka_0,\phi} = 
\bigl( R(\overline{\mathbf{z}}) \boldsymbol{u} , \boldsymbol{w} \bigr)_{\Ka_0,\phi}
 \quad\text{and}\quad
  \bigl( \boldsymbol{u}, R_j(\mathbf{z}) \boldsymbol{w} \bigr)_{\Ka_0,\phi} = 
  \bigl( R_j(\overline{\mathbf{z}}) \boldsymbol{u} , \boldsymbol{w} \bigr)_{\Ka_0,\phi}\; .
\end{equation}
\end{lemma}

\begin{proof}
To prove that the inverse $R(\mathbf{z})$ of $(Q- \mathbf{z})$ is given by the series \eqref{DefinitionOfResolvent}, we write 
$Q = Q_0 + \tilde{Q}$ and $R_0 = R_0(\mathbf{z})$ to simplify the notation. Then
\begin{equation}\label{Q-zR}
  (Q- \mathbf{z})R(\mathbf{z}) 
  = (Q_0 - \mathbf{z} + \tilde{Q}) \sum_{k=0}^\infty (-R_0 \tilde{Q})^k R_0
  = \sum_{k=0}^\infty (Q_0 - \mathbf{z}) (-R_0 \tilde{Q})^k R_0 + \sum_{k=0}^\infty \tilde{Q} (-R_0 \tilde{Q})^k R_0
\end{equation}
and, using $(R_0\tilde{Q})^k R_0 = R_0 (\tilde{Q}R_0)^k$, we get
\begin{equation*}
 \text{right hand side of}~\eqref{Q-zR} = \sum_{k=0}^\infty (-1)^k (\tilde{Q} R_0)^k + \sum_{k=0}^\infty (-1)^k (\tilde{Q} R_0)^{k+1}
  = 1.
\end{equation*}
Similar computations show that $R(\mathbf{z}) (Q-\mathbf{z}) = 1$.

To prove \eqref{SymmetryOfRj}, we use that the symmetry of $Q$ implies that for all $\boldsymbol{u}, \boldsymbol{v}\in \Ka_0$
\begin{equation}\label{symm1}
 \bigl( (Q-\bar{\mathbf{z}}) \boldsymbol{u}, R(\mathbf{z}) \boldsymbol{v} \bigr)_{\Ka_0, \phi} = \bigl( \boldsymbol{u}, \boldsymbol{v} \bigr)_{\Ka_0, \phi} = 
\bigl(R(\bar{\mathbf{z}}) (Q-\bar{\mathbf{z}}) \boldsymbol{u},  \boldsymbol{v} \bigr)_{\Ka_0, \phi}\; .
\end{equation}
Equation \eqref{symm1} in the space of formal power series yields inductively that for all $j\in \N_0/2$ and $\boldsymbol{u},\boldsymbol{v}\in \Eh_p[y]$
\[  \bigl( (Q_0 - \bar{\mathbf{z}}) \boldsymbol{u}, R_j (\mathbf{z}) \boldsymbol{v} \bigr)_{\Ka_0, \phi} = \bigl(R_j(\bar{\mathbf{z}}) (Q_0 -\bar{\mathbf{z}}) \boldsymbol{u},  \boldsymbol{v} \bigr)_{\Ka_0, \phi}\; .\]
Since $Q_0 - \bar{\mathbf{z}}: \Eh_p[y] \rightarrow \Eh_p[y]$ is bijective (see Remark \ref{SpektrumQ0}), this proves \eqref{SymmetryOfRj}.
\end{proof}

From now on, fix $E_0 \in \spec(Q_0)$. We will define a spectral projection $\Pi_{E_0}$ for $Q$ on $\Ka_0$, associated to $E_0$, 
using the operator $R(\mathbf{z})$.

By Remark \ref{SpektrumQ0}, the eigenfunctions $h_{\alpha,k}, \, \alpha \in \N_0^n,\, 1\leq k\leq \mathrm{rk}\Eh$ of $Q_0$
(see \eqref{EigenvectorsHarmonicOscillators}) form a basis in $\Eh_p[y]$. 
By \eqref{DefinitionOfResolvent}, the action of the operators $R_j(\mathbf{z})$ on $h_{\alpha,k}$
therefore determines the action of $R(\mathbf{z})$ on $\Eh_p[y]$. 
By Lemma \ref{Qj}, each $Q_j$ raises the degree of a polynomial by $2j$, thus by \eqref{Rj} and \eqref{EigenValueEquationQ0}
there exist rational functions $d^j_{\alpha,k,\beta, l}(\mathbf{z})$ on $\C$ with poles at most at $\spec (Q_0)$ such that
\begin{equation}\label{Rjhalpha}
R_j(\mathbf{z}) h_{\alpha, k} = \sum_{|\beta|\leq |\alpha| + 2j} d^j_{\alpha,k,\beta, l}(\mathbf{z}) h_{\beta, l}\; .
\end{equation}

We choose a counterclockwise oriented closed contour $\gamma\in \C\setminus \spec (Q_0)$ around $E_0$, separating $E_0$ 
from the rest of $\spec(Q_0)$, and define
\begin{equation}\label{DefPi}
  \Pi_{E_0} q :=  \frac{1}{2\pi i} \sum_{j \in \N_0/2} \hbar^j \oint_\gamma R_j(\mathbf{z}) q \, \dd \mathbf{z} \in \Ka_0\, ,\quad q\in\Eh_p[y]\, .
\end{equation}
This definition makes sense, as $R_j(\mathbf{z})q$ is holomorphic on $\C \setminus \spec(Q_0)$ in the following sense: From 
\eqref{Rjhalpha} it is clear that for fixed $q\in\Eh_p[y]$ the degree of $R_j(\mathbf{z})q$ is smaller than some $N$ for all 
$\mathbf{z} \in \C \setminus \spec(Q_0)$. Thus the range of the map $\mathbf{z} \mapsto R_j(\mathbf{z})q$ is a finite-dimensional complex vector space and as such it is holomorphic.

We extend $\Pi_{E_0}$ to $\Ka_0$ by linearity.

\begin{proposition}[The Eigenprojection]\label{propPi}
Let $m_0$ denote the multiplicity of $E_0$. Then $\Pi_{E_0}$ defined in \eqref{DefPi} is a projection in $\Ka_0$ of rank $m_0$, symmetric with respect to 
the inner product $( \,\cdot\,,\,\cdot\,)_{\Ka_0,\phi}$ and commutes with $Q$.
\end{proposition}

\begin{proof}We write $\Pi = \Pi_{E_0}$.

{\sl Symmetry:} The symmetry of $\Pi$ is a consequence of \eqref{SymmetryOfRj}. More explicitly, it follows from
\[ \Bigl( u, \oint_{\gamma}R_j(\mathbf{z})w \, \dd\mathbf{z} \Bigr)_{\Ka_0,\phi} =  
-\Bigl( \oint_{\gamma}R_j(\mathbf{z})u\, \dd\mathbf{z},  w\Bigr)_{\Ka_0,\phi}\, , \qquad u,v\in \Ka_0\, ,\]
since the orientation of $\gamma$ is reversed under complex conjugation.

$\Pi^2 = \Pi$: 
Using \eqref{DefinitionOfResolvent}, \eqref{DefPi} and the resolvent equation, this follows from standard
arguments (see \cite{klein-schwarz} or \cite{helffer-sjostrand-1} for the computation in the setting of formal power series).

rk $\Pi = m_0$: To determine the rank of $\Pi$, write 
$I = \{ (\alpha, k) \in \N_0^n \times \{1, \dots, \mathrm{rk}\,\Eh\} \mid E_{\alpha, k} = E_0 \}$. Since
\begin{equation*}
   \frac{1}{2\pi i} \int_\gamma R_0(\mathbf{z}) h_{\alpha, k}(y) \dd \mathbf{z} 
  = \begin{cases}
     h_{\alpha, k} (y)\, , & (\alpha, k) \in I \\
     0 \, ,& (\alpha, k) \notin I
     \end{cases},
\end{equation*}
we have
\begin{equation}\label{Pihalphak}
  \Pi h_{\alpha, k} (y) = 
  \begin{cases}
    h_{\alpha, k}(y) + \sum_{j \in \N/2} \hbar^j l_{j, \alpha, k}(y)\, , & (\alpha, k) \in I \\
    \sum_{j \in \N/2} \hbar^j l_{j, \alpha, k}(y)\, , & (\alpha, k) \notin I
  \end{cases}
\end{equation}
for some polynomials $l_{j, \alpha, k}\in \Eh_p[y]$ of degree less than or equal to $|\alpha|+2j$ (this follows from \eqref{Rjhalpha}).
Since the eigenfunctions $h_{\alpha,k}$ of $Q_0$ (see \eqref{EigenvectorsHarmonicOscillators}) form a basis, \eqref{Pihalphak}
 implies that
the functions $\Pi h_{\alpha,k},\, (\alpha, k) \in I$, are linearly
independent over $\C\lau$. Thus their span has
dimension $m_0$.
We now claim that the elements $\Pi h_{\alpha, k}(y)$ with $(\alpha, k) \in I$ span the range of $\Pi$. To prove this, it suffices to verify that
for all $(\beta, l)\in \N_0^n\times \{1, \ldots , \mathrm{rk}\, \Eh\}$ 
\begin{equation*}
  \Pi h_{\beta, l} (y) = \sum_{(\alpha, k) \in I} A^{\beta, l}_{\alpha, k} \Pi h_{\alpha, k}(y)
\end{equation*}
for some coefficients $A^{\beta, l}_{\alpha, k} \in \C\lau$. These coefficients $A^{\beta, l}_{\alpha, k}$, however, can be determined by an easy 
induction argument using \eqref{Pihalphak} (see e.g.\  \cite{klein-schwarz}).

$\Pi Q = Q\Pi$: Since $Q$ commutes with $R(z)$, this follows from the definition of $\Pi$.
\end{proof}

\begin{proposition}\label{orthonormprop}
For $E_{\alpha, k}$ given in \eqref{EigenvaluesHarmonicOscillator} and $E_0\in \spec (Q_0)$, let 
$I_{E_0}:= \{ (\alpha, k) \in \N_0^n \times \{1, \dots, \mathrm{rk}\,\Eh\} \mid E_{\alpha, k} = E_0 \}$ and let 
$(\alpha^j, k^j)$, $j=1, \dots m_0,$ be an 
enumeration of the elements of $I_{E_0}$. Then there exists an orthonormal basis $(\boldsymbol{b}^1, \dots, \boldsymbol{b}^{m_0})$ of $V :=  \Pi_{E_0} \Ka_0$ such that
\begin{equation}
  \boldsymbol{b}^j = h_{\alpha^j,k^j}+ \sum_{\ell \in \N/2} \hbar^{\ell} p^j_{\ell}
\end{equation}
for some polynomials $p^j_{\ell}\in\Eh_p[y]$ of degree less or equal to $|\alpha^j| + 2 \ell$.
With respect to such a basis, $Q|_V$ is represented by a Hermitian $m_0 \times m_0$ matrix $M$ of the form
\begin{equation}
  M = E_0 \cdot \mathbf{1} + \sum_{j \in \N/2} \hbar^{j} M_j
\end{equation}
where the entries of the matrices $M_j$ are in $\C\lau$ and $\mathbf{1}$ is the identity matrix.
\end{proposition}

\begin{proof}
Setting $\boldsymbol{f}^j:= \Pi_{E_0} h_{\alpha^j,k^j}$, it follows from \eqref{Pihalphak}, the definition of $(\,\cdot\, ,\,\cdot\, )_{\Ka_0,\phi}$ in \eqref{scalarprodKdef} and
the fact that the eigenfunctions $h_{\alpha^j,k^j}$ are orthonormal with respect to $(\,\cdot\, ,\,\cdot\, )_{\Ka_0,\phi}$ that 
\begin{equation}\label{fkfl}
  \bigl( \boldsymbol{f}^k, \boldsymbol{f}^\ell \bigr)_{\Ka_0,\phi} = \delta_{k\ell} + \sum_{j \in \N/2} \hbar^j \beta_{k,\ell, j}, ~~~~~ 
  \beta_{k,\ell, j} \in \C, \, \; 1 \leq k, \ell \leq m_0\; .
\end{equation}
We obtain a Hermitian $m_0 \times m_0$-matrix with entries in $\C\lau$
\begin{equation}\label{DefA}
  A := \Bigl( \bigl( \boldsymbol{f}^k, \boldsymbol{f}^\ell \bigr)_{\Ka_0,\phi} \Bigr)_{1 \leq k, \ell \leq m_0} = 
  \mathbf{1} + \sum_{j \in \N/2} \hbar^j A_j.
\end{equation}
As the series of $A$ starts with the identity matrix, the matrix 
\begin{equation}\label{DefB}
B:= A^{-1/2} = \mathbf{1} + \sum_{j \in \N/2} \hbar^j B_j
\end{equation}
is well-defined by the appropriate Taylor series and its elements are in $\C\lau$ as well. If the elements of $A$ are in
$\C(\!(\hbar)\!)$, then the same is true for the elements of $B$. Setting
\begin{equation*}
  (\boldsymbol{b}^1, \dots, \boldsymbol{b}^{m_0}) := (\boldsymbol{f}^1, \dots \boldsymbol{f}^{m_0}) B\, ,
\end{equation*}
it is straightforward to verify that $(\boldsymbol{b}^1, \dots, \boldsymbol{b}^{m_0})$ is an orthonormal basis of $V$. Furthermore, with respect to this basis, $Q$ is represented 
by the Hermitian matrix $M = B C B$ where $C$ is the matrix with entries in $\C\lau$ defined by
\begin{equation}\label{CMatrix}
  C := \Bigl(\bigl(\boldsymbol{f}^k, Q \boldsymbol{f}^\ell \bigr)_{\Ka_0,\phi}\Bigr)_{1\leq k,\ell\leq m_0} = E_0 \cdot \mathbf{1} + \sum_{j \in \N/2} \hbar^j C_{j}\, .
\end{equation}
The last equality in \eqref{CMatrix} follows from \eqref{RescaledSeriesOfH}, Lemma \ref{Qj}, Proposition \ref{propPi}, \eqref{fkfl} and the fact that
$h_{\alpha^j,k^j},\, j=1, \ldots m_0$ are the eigenfunctions of $Q_0$ for the eigenvalue $E_0$.
The symmetry of $C$ and hence $M$ follows from the symmetry of $Q$. 
The statement on the degree of the polynomials $p^j_\ell$ follows from \eqref{Pihalphak}. 
\end{proof}

\begin{corollary}[Eigendecomposition of $Q_0$] \label{EssentialCorollary}
For each $E_0 \in \spec(Q_0)$ of multiplicity $m_0$ with eigenfunctions $h_{\alpha^j,k^j}$, with $(\alpha^j,k^j) \in I_{E_0}$ as defined in Prop.\ \ref{orthonormprop}, the 
operator $Q: \Ka_0 \longrightarrow \Ka_0$ possesses $m_0$ (not necessarily distinct) eigenvalues $E_1, \dots, E_{m_0}$ of the form
\begin{equation}\label{Ej}
  \boldsymbol{E}_j= E_0 + \sum_{\ell \in \N/2} \hbar^\ell E_{j,\ell}\, , \qquad E_{j,\ell}\in\R\, , 
\end{equation}
with associated orthonormal eigenfunctions with respect to $(\,\cdot\, , \,\cdot\, )_{\Ka_0,\phi}$
\begin{equation}\label{psij}
  \boldsymbol{\psi}_j  = \sum_{\ell \in \N_0/2} \hbar^\ell \psi_{j,\ell} \in \Ka_0
\end{equation}
where $\psi_{j,\ell}\in\Eh_p[y]$ with $\deg \psi_{j,\ell} \leq d_{j,\ell}:= 2\ell + \max_{(\alpha,k)\in I_{E_0}} |\alpha|$.
\end{corollary}

\begin{proof}
This follows at once from Prop.\ \ref{orthonormprop} together with
\cite[Thm.\ A2.3]{klein-schwarz}\footnote{We recall that this result essentially is a generalization of Rellich's Theorem on the 
analyticity of eigenvalues and eigenfunctions for certain matrices with analytic coefficients (see \cite{rellich}). 
For a more abstract algebraic result (covering the present case of formal power series) see also
\cite{grater-klein}.}, which states that for any $k \in \N$, a Hermitian $m\times m$-matrix $M$ with elements 
$M_{ij}\in \C(\!(\hbar^{1/k})\!)$ has a complete eigendecomposition in $\R(\!(\hbar^{1/k})\!)$ and the associated eigenvectors can be chosen to be orthonormal.
\end{proof}

Using Lemma \ref{Qj}, we
can prove the following proposition about the absence of half integer
terms in the expansion \eqref{Ej}.

\begin{proposition}[Parity]\label{ThmHalfIntegers}
For $I_{E_0}$ as in Prop.\ \ref{orthonormprop}, we assume that all $(\alpha, k) \in I_{E_0}$ have the same parity (i.e.\ $|\alpha|$
is either even for all $(\alpha, k) \in I_{E_0}$ or odd for all 
$(\alpha, k)\in I_{E_0}$). Let $M$
denote the matrix specified in Prop.\ \ref{orthonormprop} and $E_j$
its eigenvalues given in Cor.\ \ref{EssentialCorollary}. 
Then $M_{ij}\in \C(\!(\hbar)\!)$
and $E^j\in \R(\!(\hbar)\!)$ for $1\leq i,j\leq m_0$.
\end{proposition}

\begin{proof}
Again by \cite[Thm.\ A2.3]{klein-schwarz} we know that if $M_{ij}\in\C(\!(\hbar)\!)$,
the same is true for its eigenvalues $E_j$. Thus, it suffices to
prove the statement on $M_{ij}$.
We will
change the notation during this proof to $\boldsymbol{f}_{\alpha,k}:=\Pi_{E_0} h_{\alpha,k}$ and
$C_{\alpha,k,\beta, \ell}$ for $(\alpha,k), (\beta, \ell)\in I_{E_0}$ and $C$ given in \eqref{CMatrix}. \\
We start proving that
$\bigl(\boldsymbol{f}_{\alpha,k}, \boldsymbol{f}_{\beta, \ell}\bigr)_{\Ka_0, \phi}\in \C(\!(\hbar)\!)$. By the definition
\eqref{DefPi} of $\Pi_{E_0}$, we have
\begin{equation}\label{falphaj}
\boldsymbol{f}_{\alpha,k} = \sum_{j\in \N_0/2} \hbar^j f_{\alpha,k,j} \quad \text{with} \quad 
\boldsymbol{f}_{\alpha,k,j} = \frac{1}{2\pi i}\oint_{\Gamma} R_j(\mathbf{z})h_{\alpha,k}\, d\mathbf{z}\, .
\end{equation}
By \eqref{Rj}, $R_j(\mathbf{z})$ is determined by $Q_{j_\ell}$ with $\sum j_\ell = j$ and $R_0(\mathbf{z})$. Since
by Lemma \ref{Qj} the operator $Q_{j_\ell}$
changes the parity of a polynomial in $\Eh[y]$ by the factor
$(-1)^{2j_\ell},\, j_\ell \in\N_0/2,$ (and raises its degree by $2j_\ell$), we can conclude
that $R_j(\mathbf{z})$ changes the parity by
$(-1)^{2j}$. Using that, for each $1 \leq k \leq \mathrm{rk}\Eh$, the parity of $h_{\alpha,k}$ is
given by $(-1)^{|\alpha|}$, it follows that the parity of $f_{\alpha,k, j}$ is given by $(-1)^{|\alpha|+2j}$. 
By \eqref{skpk} we have
\begin{equation}\label{Anull}
\bigl(\boldsymbol{f}_{\alpha,k}, \boldsymbol{f}_{\beta,\ell}\bigr)_{\Ka_0, \phi} 
 =\sum_{n\in\N_0/2} \hbar^n
   \sum_{\natop{j, m, s, r\in\N_0/2}{j+m+s+r=n}}\int_{\R^n} \gamma_r[{f}_{\alpha,k, j}, {f}_{\beta, \ell, m}](y)
   \omega_s(y) e^{-\langle y, \Lambda y\rangle}\, dy \, .
\end{equation}
We shall show that for $2n$ odd (and thus for $n$ half-integer)
each summand vanishes. For fixed $j,m,s,r$, the integral will vanish
if the entire integrand is odd. According to Prop.\ \ref{rescaleProp2}, the
parity of $\omega_s$ is $(-1)^{2s}$. Moreover, as described below \eqref{Rundgammatilde}, the parity of the term
$\gamma_r[f_{\alpha,k, j}, f_{\beta, \ell, m}]$ is given by $(-1)^{|\alpha|+2j + |\beta| + 2m + 2r}$.
Since the exponential term is even, the integral on the right hand side of \eqref{Anull} therefore
vanishes if $(|\alpha| + 2j + |\beta| + 2m + 2r + 2s)$ is odd. Since by
assumption $\alpha$ and $\beta$ have the same parity, $|\alpha| +
|\beta|$ is even. Thus the integral vanishes if $2(j+m + s + r)= 2n$ is odd which occurs
if $n$ is half-integer. This shows that
$A_{\alpha,k, \beta, \ell} = \bigl(\boldsymbol{f}_{\alpha,k}, \boldsymbol{f}_{\beta,\ell}\bigr)_{\Ka_0, \phi}\in \C (\!(\hbar)\!)$  (see \eqref{DefA}) and thus the same is
true for $B_{\alpha,k,\beta, \ell}$ given in \eqref{DefB}. \\
Since $M=BCB$ for $C$ given in \eqref{CMatrix} as described in the proof of Prop.\ \ref{orthonormprop}, it remains to show that 
the elements of $C$ are in $\C(\!(\hbar)\!)$ where 
\begin{multline}\label{Calphabeta}
C_{\alpha,k,\beta, \ell} = \bigl(\boldsymbol{f}_{\alpha,k}, Q \boldsymbol{f}_{\beta,\ell}\bigr)_{\Ka_0, \phi} \\ = 
\sum_{n\in\N_0/2}\hbar^n
   \sum_{\natop{j, m, s, r, q\in\N_0/2}{j+m+s+r+q=n}}\int_{\R^n} \gamma_r[f_{\alpha,k, j}, Q_q f_{\beta, \ell, m}](y)
   \omega_s(y) e^{-\langle y, \Lambda y\rangle}\, dy \, .
\end{multline}
By Lemma \ref{Qj}, the operator $Q_q$ changes the parity by $(-1)^{2q}$. Therefore, it follows from the discussion above \eqref{Anull} that
the integral on the right hand side of \eqref{Calphabeta}
vanishes if $j + m + s + r + q = n$ is half integer. This proves the proposition.
\end{proof}

Now we come back to the Taylor expansion of $H_{\phi,\hbar}$ with respect to our original not rescaled chart $x$. 

\begin{theorem}[Eigendecomposition of $\tau_p(H_{\phi, \hbar})$]\label{PropEigenmitx}
Let $\hbar E_0$ be an eigenvalue of multiplicity $m_0$ of the local 
harmonic oscillator $H_{p, \hbar}$ from Def.\ \ref{DefLocalHarmonicOscillator}. 
Then the operator $\tau_p \bigl(H_{\phi, \hbar}\bigr)$ on $\Ka$ as in \eqref{LocalFormHPhi} has a 
system of $m_0$ eigenfunctions $\boldsymbol{\hat{a}}_j\in \Ka$, $j=1. \ldots m_0,$ that are orthonormal with respect to $(\,\cdot\,, \,\cdot\,)_{\Ka, \phi}$ 
and that are of the form
\begin{equation}\label{eigenfctonK}
 \boldsymbol{\hat{a}}_j = \hbar^{-K}\sum_{k\in \N_0/2} \hbar^k \boldsymbol{a}_{j,k} \quad\text{with}\quad 
\boldsymbol{a}_{j,k} = \sum_{\beta\in\N^n} {a}_{j,k,\beta} x^\beta \in   \Eh_p[[x]]
\end{equation}
where $K = \max_{\alpha, k \in I_{E_0}}|\alpha|/2$ for $I_{E_0}$ as given in Prop.\ \ref{orthonormprop} and the lowest order monomial in 
$\boldsymbol{a}_{j,k}$ is of degree $\max\{2(K-k), 0\}$. The associated eigenvalues are 
\begin{equation}\label{eigenvalueonK}
 \hbar \boldsymbol{E}_j =  \hbar \Bigl( E_0 +  \sum_{k \in \N/2} \hbar^k E_{j,k} \Bigr) \, , \qquad E_{j,k}\in\R\, .
\end{equation}
If $|\alpha|$ is even (or odd resp.) for all $(\alpha, \ell)\in I_{E_0}$, then all half integer terms (or integer terms respectively) in the expansion 
\eqref{eigenfctonK} vanish, and in both cases, all half integer terms in the expansion \eqref{eigenvalueonK} vanish.
\end{theorem}

\begin{proof}
As discussed in Remark \ref{RemarkOnH0AndQ0}, $E_0$ is an eigenvalue of multiplicity $m_0$ of $Q_0$. Thus 
we derive the eigenfunctions $\boldsymbol{\hat{a}}_j\in \Ka$ of $\tau_p\bigl(H_{\phi, \hbar}\bigr)$ by rescaling the eigenfunctions $\boldsymbol{\psi}_j\in \Ka_0$ 
of $Q$ given in \eqref{psij}, explicitly 
\[
\boldsymbol{\hat{a}}_j = R^{-1} \boldsymbol{\psi}_j = R^{-1} \sum_{\ell\in \N_0/2}\hbar^\ell \sum_{\natop{\beta\in \N_0^n}{|\beta| \leq d_{j\ell}}} 
\psi_{j,\ell,\beta} y^\beta =
 \sum_{\ell\in \N_0/2} \sum_{\natop{\beta\in \N_0^n}{|\beta| \leq d_{j\ell}}}\hbar^{\ell-\frac{|\beta|}{2}} \psi_{j,\ell,\beta} x^\beta \] 
for $d_{j,\ell} =2 \ell + \max_{(\alpha,k)\in I_{E_0}}  |\alpha| $ (see Corollary \ref{EssentialCorollary}). 
Thus $\ell-\frac{|\beta|}{2} \geq - \max_{(\alpha,k)\in I_{E_0}}{ |\alpha|}/{2} =: - K$. Moreover, setting $k=\ell - {|\beta|}/{2} + K$, we get 
$|\beta|/2\geq K-k$ and 
\[
\boldsymbol{\hat{a}}_j  =  \hbar^{-K}\sum_{k\in \N_0/2} \hbar^k \sum_{\natop{\beta\in\N_0^n}{|\beta|\geq 2(K-k)}} \psi_{j,|\beta|/2-K+k, \beta} x^\beta =:  
\hbar^{-K}\sum_{k\in \N_0/2} \hbar^k \boldsymbol{a}_{j,k} \, .
\]
Thus ${a}_{j,k,\beta} = \psi_{j,|\beta|/2-K+k, \beta}$ and the lowest degree of $\boldsymbol{a}_{j,k}$ is given by $\max \{ 2(K-k), 0\}$.

The orthonormality of the eigenfunctions $\boldsymbol{\hat{a}}_j = R^{-1} \boldsymbol{\psi}_j$ follows at once from Corollary \ref{EssentialCorollary}
together with Definition \ref{DefProdK_0}.

We now consider the case that all $|\alpha|$ are even (or odd respectively) for $(\alpha, \ell)\in I_{E_0}$.
By Corollary \ref{EssentialCorollary}, Prop.\ \ref{ThmHalfIntegers} and  \cite[Thm.\ A2.3]{klein-schwarz}, we can write any eigenfunction 
$\boldsymbol{\psi}\in \Pi_{E_0}\Ka_0$ of $Q$ as linear combination of 
$\Pi_{E_0} h_{\alpha,\ell}\, , \; (\alpha, \ell)\in I_{E_0},$ 
with coefficients $\lambda_{\alpha,\ell}\in \C(\!(\hbar)\!)$. Thus, using the notation in the proof on Prop.\ \ref{ThmHalfIntegers}, by \eqref{falphaj} we 
explicitly get
\begin{equation}\label{f1} 
\boldsymbol{\psi} = \sum_{(\alpha,\ell)\in I_{E_0}} \sum_{k\in \N_0} \hbar^k \lambda_{\alpha,\ell,k} \sum_{s\in\N_0/2}\hbar^s f_{\alpha,\ell,s}\; . 
\end{equation}
As discussed below \eqref{falphaj}, the polynomials $f_{\alpha,\ell,s}$ are of degree $|\alpha| + 2s$ in $y$ and have the parity 
$(-1)^{|\alpha| + 2s}$. Thus rescaling explicitly gives 
\begin{equation}\label{f2} 
R^{-1} f_{\alpha,\ell,s} = R^{-1} \!\!\sum_{\natop{r\in\N_0}{r\leq |\alpha|/2 + s}} \sum_{\natop{\beta\in \N_0^n}{|\beta| = |\alpha| + 2s + 2r}}
\!\!f_{\alpha,\ell,s,\beta} y^\beta =  \sum_{\natop{r\in\N_0}{r\leq |\alpha|/2 + s}} \!\hbar^{|\alpha|/2 + s - r }\! 
\sum_{\natop{\beta\in \N_0^n}{|\beta| = |\alpha| + 2s - 2r}}\!\!
f_{\alpha,\ell,s,\beta}x^\beta \, .
 \end{equation}
Inserting \eqref{f2} into \eqref{f1} and setting $m=k +2s - r\in\Z$ gives the expansion
\[ R^{-1} \boldsymbol{\psi} =  \sum_{(\alpha,\ell)\in I_{E_0}} \sum_{m\in\Z, m\geq -|\alpha|/2} \hbar^{|\alpha|/2 + m } 
p_{\alpha,\ell, m} \in \Ka \; . \]
Thus if $|\alpha|$ is odd (or even respectively), then $|\alpha|/2 + m$ is half-integer (or integer resp.).
So if one of these assumptions is true for all $(\alpha,\ell)\in I_{E_0}$, there remain no integer terms (or half-integer terms respectively) in the
expansion of $R^{-1} \boldsymbol{\psi}$. Since the transition to $\boldsymbol{\hat{a}}_j$ is just a reordering, this is also true for 
$\boldsymbol{\hat{a}}_j$.  

The statement on the eigenvalues follows at once from Corollary  \ref{EssentialCorollary}, the definition of $Q$ in \eqref{RescaledSeriesOfH}
and Prop.\ \ref{ThmHalfIntegers}.
\end{proof}

\section{Proof of Theorem \ref{Theorem1}}\label{Kapitel4}

Given Setup \ref{setup1}, we fix an admissible pair $(U, \phi)$ and an eigenvalue $\hbar E_0$ of $H_{p,\hbar}$ of multiplicity $m_0$. For  $j=1, \ldots, m_0$, 
let 
$\boldsymbol{\hat{a}}_j$ and $\boldsymbol{E}_j$ be the associated orthonormal eigenfunctions and eigenvalues of $\tau_p\bigl(H_{\phi, \hbar}\bigr)$ as given in 
Thm.\ \ref{PropEigenmitx}. 
By Corollary \ref{BorelTheorem}, for each $k\in \N_0/2$ there exist sections $\tilde{a}_{j,k} \in \Gamma^\infty(U, \Eh)$ such that
$\tau_p\bigl( \tilde{a}_{j,k}\bigr) = \boldsymbol{a}_{j,k}$. 

Then by Theorem \ref{PropEigenmitx} we have
\begin{equation} \label{RestEquation}
  \Bigl(H_{\phi, \hbar} - \hbar\bigl( E_0 +  \sum_{k \in \N/2} \hbar^k E_{j,k} \bigr)\Bigr) \sum_{k\in \N_0/2} \hbar^k \tilde{a}_{j,k} = \hbar\sum_{k\in \N_0/2} 
\hbar^k r_{j,k}
\end{equation}
in the sense of formal series in $\hbar^{1/2}$ with coefficients in $\Gamma^\infty(U, \Eh)$, where $r_{j,k}\in \ker \tau_p \cap \Gamma^\infty_c (U, \Eh)$, 
i.e.,  $r_{j,k}$ vanishes to infinite order at $p\in M$ for all $k\in \N_0/2, j=1, \ldots m_0$. 

We now want to modify the sections $\tilde{a}_{j,k}$ such that \eqref{RestEquation} is solved with zero on the right hand side. 
By Remark \ref{RmkTransportEquations}, this is the case for a series $\sum \hbar^k a_{j, k}$ if and only if the coefficients 
$a_{j,k}$ solve the transport equations \eqref{TransportEquations}. Our terms $\tilde{a}_{j,k}$ solve the transport equations {\em almost}, up to the error 
terms $r_{j,k}$ which vanishes to infinite order at $0$. To get rid of these terms, we use the following theorem.

\begin{theorem}[Flat Solutions] \label{TheoremFlatSolutions}
Let $X \in \Gamma^\infty(U, TM)$ be a vector field vanishing at $p \in U$ such that the eigenvalues of its linearization $\nabla X|_p$ at $p$ all have positive 
real part and let $A$ be an endomorphism field of the vector bundle $\Eh$. Assume that $U$ is star-shaped around $p$ with respect to $X$. Then for each 
section $r \in \ker \tau_p\cap \Gamma^\infty (U, \Eh)$,
there exists a unique section $\eta \in \ker \tau_p\cap \Gamma^\infty (U, \Eh)$ 
solving the differential equation
\begin{equation}\label{allform}
  (\nabla_X^\Eh + A)\, \eta = r.
\end{equation}
\end{theorem}

For the scalar case, a proof of Theorem \ref{TheoremFlatSolutions} is contained in \cite{helffer-sjostrand-1} (and amplified e.g. in \cite{dima}). In \cite{helffer-sjostrand-4} Theorem \ref{TheoremFlatSolutions} was used for the bundle case
to analyze the Witten complex 
(without giving a complete formal proof). In fact, such a
proof only needs minor modifications compared to the scalar case (basically due to the fact that the non-autonomous ODE 
$\dot{x} = A(t) x$ has no solution in terms of elementary functions for $A(t)$ a matrix, making the variation of constants
formula slightly less explicit). A complete proof of Theorem \ref{TheoremFlatSolutions} (in the case $M=\R^n$) is contained e.g. in 
\cite[Theorem A.1, Step 3]{giacomo} or (for general $M$) in \cite{matthias}.\\

For  $j=1, \dots, m_0$, we now look for a solution $\boldsymbol{\eta}_j := \sum_{k\in \N_0/2} \hbar^k \eta_{j,k}$ of \eqref{RestEquation} with $\eta_{j,k}\in\ker \tau_p\cap \Gamma^\infty (U, \Eh)$.
The first equation is 
\begin{equation} \label{aj0}
\bigl(\nabla_X^\Eh + A\bigr)\eta_{j,0} = r_{j,0}~~~~~~\text{with}\quad X = 2 \grad \phi~~~~\text{ and }~~~~A = W + \Delta \phi - E_0\; . 
\end{equation}
It follows from Setup \ref{setup1} that $X$ and $A$ fulfill the assumptions given in Thm \ref{TheoremFlatSolutions} and since $(U, \phi)$ is 
admissible, $U$ is star-shaped around $p$ with respect to $X$ by Definition \ref{DefAdmissible}.
Thus \eqref{aj0} has a unique solution
$\eta_{j,0}\in  \ker \tau_p\cap \Gamma^\infty (U, \Eh)$ and $a_{j,0}:= \tilde{a}_{j,0} - \eta_{j,0}$ solves $\bigl(\nabla_X^\Eh + A\bigr)a_{j,0} = 0$.
We now proceed inductively. 
Assume that the equations of order $1\leq \ell \leq k+1/2$ are solved by sections $\eta_{j, \ell-1}\in  \ker \tau_p\cap \Gamma^\infty (U, \Eh)$. 
Then for $X$ and $A$ as given in \eqref{aj0}, the equation of order $k+1 \in \N/2$ is given by
\begin{equation}\label{etajk}
  (\nabla_X^\Eh + A)\eta_{j, k}   = r_{j, k}  -  L \eta_{j,k-1} + \sum\nolimits_{i=1/2}^k E_{j,i}\, \eta_{j, k-i}\, .
\end{equation}
Now as the right hand side of \eqref{etajk} is known and flat at $p$, it can be considered as inhomogeneity $r$ and by Thm.\ \ref{TheoremFlatSolutions}, there exists a 
unique solution $\eta_{j,k}\in \ker \tau_p\cap \Gamma^\infty (U, \Eh)$ of \eqref{etajk}.

Setting $a_{j,k} := \tilde{a}_{j,k} - \eta_{j,k}\in \Gamma^\infty (U, \Eh)$ for $j=1, \dots, m_0$ and $k\in \N_0/2$, it follows that 
\[ \boldsymbol{a}_j :=  \hbar^{-K}\sum_{k\in\N_0/2} \hbar^k a_{j,k} \] 
solves 
\begin{equation*}
  \Bigl(H_{\phi, \hbar} - \hbar \boldsymbol{E}_j \Bigr) \boldsymbol{a}_j = 0
\end{equation*}
in the sense of asymptotic series in $\hbar^{1/2}$ with coefficients in $\Gamma^\infty(U, \Eh)$, thus
Property 1 of Thm \ref{Theorem1} is shown. 

To prove Property 2 of Thm \ref{Theorem1}, we first remark that for any cut-off function $\chi\in \Gamma^\infty_c (U, [0,1])$ with $\chi\equiv 1$ in a 
neighborhood of $p$, we have $\chi a_{j,k}\in \Gamma^\infty_c (U, \Eh)$ for $j=1, \ldots m_0$ and $k\in \N_0/2$. Moreover, since $\chi$ and 
$\eta_{j,k}$ are in $\ker\tau_p$, we have
by construction $\tau_p (\chi a_{j,k}) = \tau_p(\tilde{a}_{j,k}) = \boldsymbol{a}_{j,k}$, $j=1, \ldots m_0, k\in \N_0/2$. Thus 
by Thm.\ \ref{PropEigenmitx}, \eqref{scalarprodmitI} and \eqref{extendI} we have in $\C(\!(\hbar^{1/2})\!)$ for all
$i,j=1, \ldots m_0$
\begin{align*} 
 \delta_{ji} &= \bigl(\boldsymbol{\hat{a}}_j, \boldsymbol{\hat{a}}_i\bigr)_{\Ka, \phi} = 
 \Bigl(\hbar^{-K}\!\sum_{k\in\N_0/2} \hbar^k \boldsymbol{a}_{j,k}, 
\hbar^{-K}\!\sum_{\ell\in\N_0/2} \hbar^\ell \boldsymbol{a}_{i,\ell}\Bigr)_{\Ka, \phi}\\
&=
 \Bigl(\hbar^{-K}\!\sum_{k\in\N_0/2} \hbar^k \tau_p \bigl(\chi a_{j,k}\bigr), 
\hbar^{-K}\!\sum_{\ell\in\N_0/2} \hbar^\ell \tau_p\bigl(\chi a_{i,\ell}\bigr)\Bigr)_{\Ka, \phi} \\
&=  \mathcal{I}\Bigl( \gamma \bigl[\hbar^{-K}\!\sum_{k\in\N_0/2} \hbar^k\chi a_{j,k}, \hbar^{-K}\!\sum_{\ell\in\N_0/2} \hbar^\ell \chi a_{i,\ell}\bigr]
\Bigr) \\
 &=  \mathcal{I}\Bigl( \gamma  \bigl[\chi \boldsymbol{a}_j, \chi \boldsymbol{a}_i\bigr]\Bigr)\; .
\end{align*}

The statement about the vanishing of terms of half-integer or integer order, depending on the parity of $|\alpha|$, follows from the analogous result in 
Thm.\ \ref{PropEigenmitx}.

\begin{remark}
By similar arguments as for the method of stationary phase, we get the following analytic statement from \eqref{RestEquation}: For each open set $U^\prime \subset \subset U$, for each $N \in \N/2$ and for each $\hbar_0>0$, there exists a constant $C>0$ such that
\begin{equation} 
  \Bigl|\Bigl(H_{\hbar} - \hbar\bigl(E_0 - \sum\nolimits_{k=1/2}^N \hbar^k E_{j,k}\bigr)\Bigr) e^{-\phi/\hbar}
 \sum\nolimits_{k = 0}^N \hbar^k \tilde{a}_{j,k} \Bigr| \leq C \hbar^{N+K+3/2}.
\end{equation}
uniformly on $U^\prime$. The extra factor $\hbar^K$ is present on the right hand side because the terms $\tilde{a}_{j, k}$ vanish to order at least $2K - k$ by Thm.\ \ref{PropEigenmitx}. However, to get the stronger result
\begin{equation}\label{Equation77}
  \Bigl| \Bigl( H_{\phi, \hbar} - \hbar\bigl(E_0 - \sum\nolimits_{k=1/2}^N \hbar^k E_{j,k}\bigr) \Bigr) \sum\nolimits_{k=0}^N \hbar^k a_{j, k} \Bigr| 
\leq C \hbar^{N+K +3/2}
\end{equation}
uniformly on $U^\prime$ (which is equivalent to property 1 of Thm.\ \ref{Theorem1}), we need Thm. \ref{TheoremFlatSolutions}.
\end{remark}

\end{document}